\documentclass[12pt, border=10pt, a4paper]{article}
\usepackage[utf8]{inputenc}
\usepackage{tikz}
\usepackage{subcaption}
\usepackage[affil-it]{authblk}
\usepackage{multirow}
\usepackage{comment}
\usepackage{lineno}
\usepackage{lipsum}
\usepackage{setspace}
\usepackage{titling}
\usepackage{subcaption}
\usepackage{algorithm}
\usepackage[noend]{algpseudocode}
\usepackage{epigraph}
\usepackage{graphicx}
\usepackage{color}
\usepackage[margin=1in]{geometry}

\usepackage{algorithm}
\usepackage[noend]{algpseudocode}
\usepackage{mdframed}
\usepackage{xr-hyper}
\usepackage[colorlinks,citecolor=northwestern-purple,urlcolor=chicago-maroon,linkcolor=chicago-maroon]{hyperref}
\hypersetup{
    colorlinks,
    linkcolor={chicago-maroon},
    filecolor={cornell-red},
    urlcolor={northwestern-purple},
    citecolor={northwestern-purple}
}
\usepackage{pgfpages}\usepackage{epsfig}\usepackage{multirow}\usepackage{lscape}\usepackage{epic}
\usepackage{pgfplots}
\usetikzlibrary{arrows, arrows.meta, automata, positioning, decorations, decorations.pathreplacing, calc,matrix,pgfplots.fillbetween,pgfplots.groupplots}
\usepackage[round]{natbib}
\usepackage{amsmath}
\usepackage{amssymb}
\usepackage{amsthm}
\usepackage{subcaption}
\usepackage{etoolbox}
\usepackage{lipsum}
 \usepackage{fullpage}
\usepackage[nameinlink]{cleveref}
\usepackage[normalem]{ulem}
\usepackage{footmisc}
\usepackage{physics}

\usepackage{enumitem}

\makeatletter
\newcommand*{\addFileDependency}[1]{
  \typeout{(#1)}
  \@addtofilelist{#1}
  \IfFileExists{#1}{}{\typeout{No file #1.}}
}
\makeatother

 \setlist[enumerate]{leftmargin=0cm,rightmargin=0cm,noitemsep, topsep=2pt}

\definecolor{clemson-orange}{RGB}{234,106,32}
\definecolor{chicago-maroon}{RGB}{128,0,0}
\definecolor{northwestern-purple}{RGB}{82,0,99}
\definecolor{cornell-red}{RGB}{179,27,27}
\definecolor{sauder-green}{RGB}{171,180,0}
\definecolor{gray}{RGB}{192,192,192}
\definecolor{lawngreen}{RGB}{0,250,154}

\makeatletter
\def\BState{\State\hskip-\ALG@thistlm}
\makeatother
\allowdisplaybreaks
\newcommand{\bb}{\mathbb}
\newcommand{\Z}{{\bb Z}}
\newcommand{\N}{{\bb N}}

\newcommand{\E}{\bb E}

\theoremstyle{definition}
\newtheorem{theorem}{Theorem}

\newtheorem{corollary}{Corollary}

\theoremstyle{definition}

\makeatletter
\patchcmd{\@addmarginpar}{\ifodd\c@page}{\ifodd\c@page\@tempcnta\m@ne}{}{}
\makeatother
\crefname{assumption}{Assumption}{Assumptions}
\crefname{lemma}{Lemma}{Lemmas}
\crefname{theorem}{Theorem}{Theorems}
\crefname{corollary}{Corollary}{Corollaries}
\crefname{proposition}{Proposition}{Propositions}
\crefname{claim}{Claim}{Claims}
\crefname{procedure}{Procedure}{Procedures}
\crefname{algorithm}{Algorithm}{Algorithms}
\crefname{figure}{Figure}{Figures}
\crefname{remark}{Remark}{Remarks}
\crefname{section}{Section}{Sections}
\crefname{procedure}{Procedure}{Procedures}
\crefname{example}{Example}{Examples}
\crefname{definition}{Definition}{Definitions}
\crefname{table}{Table}{Tables}
\crefname{equation}{}{}
\crefname{enumi}{}{}
\crefname{conjecture}{Conjecture}{Conjectures}
\crefname{step}{Step}{Steps}
\def \a{\alpha}

\def \m{\mu}
\def \l{\lambda}

\def \t{\theta}
\def \T{\Theta}

\def \E{\mathbb{E}}

\def \l{\lambda}

\def \M{\mathcal M}

\def \r{\rho}
\def \V{\mathcal V}

\def\H{H}

\def \ll{\lower1.6truept\hbox{${{\scriptstyle =\atop \scriptstyle <}}$}}
\def \gl{\lower1.6truept\hbox{${{\scriptstyle >\atop \scriptstyle
=}\atop{\scriptstyle <}}$}}
\def \lg{\lower1.6truept\hbox{${{\scriptstyle <\atop \scriptstyle =}\atop
{\scriptstyle >}}$}}
\def\g{\gamma}
\def\l{\lambda}

\def\bp{\mathbf{p}}

\def\fcfs{\text{\textsf{\tiny  FCFS}}}
\def\lcfs{\text{\textsf{\tiny  LCFS}}}
\def\siro{\text{\textsf{\tiny  SIRO}}}
\def\ntu{\text{\textsf{\tiny  NTU}}}
\def\tu{\text{\textsf{\tiny  TU}}}
\def\bv{\mathbf{v}}

\thanksmarkseries{arabic}

\begin{document}

\title{\textbf{Dynamic Market Design\thanks{This is a survey prepared for an invited lecture at the 2025 World Congress of the Econometric Society.  I am grateful to Dirk Bergemann, Andy Choi, and Sangjun Park, for their comments.
}}}
\author{Yeon-Koo Che\thanks{
Department of Economics, Columbia University, USA. Email: \href{mailto: yeonkooche@gmail.com}%
{\texttt{yeonkooche@gmail.com}}.}}
\date{\today \endgraf }
\maketitle

\begin{abstract}
Classic market design theory is rooted in static models where all participants trade simultaneously. In contrast, modern platform-mediated digital markets are fundamentally dynamic, defined by the asynchronous and stochastic arrival of supply and demand. This chapter surveys recent work that brings market design to this dynamic setting. We focus on a methodological framework that transforms complex dynamic problems into tractable static programs by analyzing the long-run stationary distribution of the system. The survey explores how priority rules and information policy can be designed to clear markets and screen agents when monetary transfers are unavailable, and, when they are available, how queues of participants and goods can be managed to balance intertemporal mismatches of demand and supply and to spread competitive pressures across time.

 \vskip0.2cm \noindent \textbf{%
JEL Classification Numbers}: C78, C61, D47, D83, D61 \newline
\end{abstract}



\newpage

\doublespacing
\section{Introduction}

The basic workhorse models of market design are rooted in two great traditions. For markets mediated by monetary prices, or Transferable Utility (TU) markets, the foundational theories of auctions by \cite{myerson1981} and \cite{milgrom1982theory} and matching theories of \cite{shapley1971assignment} and \cite{becker1973theory} provide our core understanding. For markets coordinated without direct transfers, or Non-Transferable Utility (NTU) markets, the canonical matching model of \cite{gale1962college} is the foundational paradigm. These seminal models, despite their power, reflect a classic market setting where all buyers and sellers must be present simultaneously at the same location to execute trades.

The advances in information technology and the internet have fundamentally broken down these physical and temporal barriers. The defining feature of the resulting modern marketplaces is that participants---buyers and sellers, or agents and goods---arrive over time \textbf{asynchronously and stochastically}. A typical transaction today is mediated by a digital platform that must manage unpredictable, often mismatched flows of supply and demand. These platforms do not merely serve as passive meeting grounds; they actively structure the market by matching or recommending participants to one another, in many cases deciding who must wait for a match and for how long.

Examples of such dynamic markets are ubiquitous. In the NTU setting, asynchronous and stochastic arrivals figure prominently in kidney exchange, public housing,  ride-hailing,\footnote{While ride-hailing platforms often employ dynamic (surge) pricing, the actual allocation or dispatch of a specific rider to a driver is typically mediated by proximity and arrival order rather than a price-based auction for each match. Rare exceptions exist, such as the Czech platform Liftago, which uses explicit auctions to mediate dispatch; however, for the vast majority of platforms, the allocation decision remains better classified as an NTU case.} and professional job matching. In the TU setting, examples include the sale of cloud computing capacity, where demand for server time arrives and supply is freed up as old jobs complete; gig platforms like TaskRabbit, which face unpredictable availability of both service providers and customer requests; and blockchains, which must allocate randomly arriving transactions to blocks that are themselves generated at random intervals.

The goal of this paper is to survey recent works that bring the classical agenda of market and mechanism design up to date by placing it in this more realistic dynamic setting. The relevant literature is already voluminous and growing rapidly, and we do not aim to be exhaustive. Instead, we focus on a strand of recent literature, developed largely around the works of \cite{baccara2020optimal}, \cite{che2021optimal}, \cite{madsen2025collective}, and \cite{che2025optimal}, that highlights a particular methodological framework capable of collapsing the complex dynamics of these markets into a tractable, static linear programming problem.

The workhorse framework we develop borrows heavily from queueing theory. For illustrative clarity, we will ground the survey in a simple M/M/1 environment, where buyers arrive at a Poisson rate $\lambda$ and items (or service completions) arrive at a Poisson rate $\mu$.\footnote{In Kendall's notation, M/M/1 denotes a system with Markovian (Poisson) arrivals, Markovian (exponential) service times, and a single server.} While much of the analysis can be generalized, this setting is ideal for conveying the core ideas. We consider a general class of mechanisms, called Positive Recurrent Regenerative Mechanisms (PRRMs), which induce a process for the system’s state that is positive recurrent and regenerative. This class is general enough to be without loss for most problems of interest and nests, as a special case, the simpler class of Markovian mechanisms that condition policies on payoff-relevant states such as the size and types of agents in the queue. This generality is warranted, as the optimal mechanisms we identify are often non-Markovian. Despite their generality, PRRMs retain crucial tractability because they admit a unique stationary distribution, allowing the designer’s objective and the agents' incentives to be evaluated in the long-run steady state.

The central methodology, therefore, is to transform the dynamic design problem into a static optimization problem: choosing a stationary distribution over the space of queue states that maximizes the designer's objective, subject to the constraint that this distribution must be implementable by some feasible PRRM. This presents two primary challenges. First, it is not immediately obvious how to characterize the set of feasible stationary distributions. We show how this can be done using a Border-inspired characterization of reduced-form allocations, adapted to a dynamic setup. Second, the space of queue states can be forbiddingly complex—for instance, including the reported types of all agents in the queue—making the dimensionality of the stationary distribution intractable. We demonstrate how this dimensionality can be reduced to make the analysis tractable.

This survey is organized into two parts: NTU and TU markets.

First, we analyze the NTU setting, where goods and services are allocated without monetary transfers. While other non-monetary mechanisms exist, we focus on the most common: waiting in line, or ``queueing.''\footnote{Here, we use the term queueing generically to mean any methods that involve agents waiting in line, and not necessarily to mean a particular queueing rule such as first-come-first-served.} Though queueing can play a role similar to pricing in clearing markets and screening types, there are noteworthy distinctions. In a competitive market, a decision to pay a monetary price has no direct externality on others. In contrast, a decision to join a queue can impose a significant externality, and the nature of this externality depends on three key factors: the admissions control policy (how entry and exit are regulated), the queueing rule (how priority is assigned), and the information environment (what agents know about the state of the queue). Moreover, the decision is fundamentally dynamic, as agents are typically free to abandon the queue at any time. This suggests that special attention must be paid to how queueing incentives are managed to serve the designer’s objective.

We will survey several key papers that study this problem under complete information, summarizing the findings of \cite{naor1969regulation}, \cite{hassin1985}, and \cite{leshnoAER} on the relative merits of First-Come, First-Served (FCFS), Last-Come, First-Served (LCFS), and Service-In-Random-Order (SIRO). We then discuss how \cite{che2021optimal} applies the steady-state framework to analyze optimal queue design, showing that FCFS re-emerges as optimal when the designer has a complete toolkit. Finally, we discuss the implications of using queueing to screen agents with heterogeneous types. While waiting can function like a price, it entails wasteful social costs, a crucial distinction from monetary transfers. This trade-off implies that a welfare-motivated designer may be willing to sacrifice allocative efficiency to avoid these costs, sometimes resulting in complete pooling and random allocation as the optimal policy.

Second, we turn to the TU model, where monetary transfers can be used without restriction. Given the availability of transfers, the wasteful queueing of the NTU world need not be relied upon for screening or market clearing. Instead, the central question is how to balance competitive forces across time optimally. This adds an important new dimension to the static design problem of \cite{myerson1981}: the designer must not only allocate goods/services optimally among currently present agents but also store buyers or goods optimally in a queue for potential future allocation. This involves dynamically managing the entry and exit of buyers into and out of queues based on their types, as well as managing an inventory of goods. We will then summarize the main results of \cite{che2025optimal}, which characterizes the optimal dynamic auction in this setting.

A final section concludes by discussing some topics and literature not covered in the survey and suggesting directions for future research.

\section{Non-Transferable Utility Model}

\label{sec:model}

\paragraph*{Primitives.} 

We consider a continuous-time model in which a platform/designer allocates goods or services arriving at the platform to buyers who also arrive at the platform over time.    At each instant $t\in [0, \infty)$, buyers arrive at a Poisson rate of $\lambda>0$ and homogeneous goods---or  a firm offering the goods--- arrive at a Poisson rate of $\mu>0$. In the context of the service center, the arrival of goods can be interpreted as the arrival of ``service completions," which would be enjoyed by a buyer if they have been receiving service, or would be wasted if there is no buyer receiving service.

Initially, we assume that buyers are also homogeneous; they value a fixed surplus $v>0$ from receiving the good/service, and they incur $c\in (0,v)$ per unit of time they wait in the queue. There is no discounting.\footnote{It is customary to assume in the queueing literature that the only discounting involves linear waiting costs, as we assume. There are two reasons.  First, exponential discounting introduces risk-loving time preferences, which are often contrary to what authors believe customers exhibit (e.g., risk-averse time preferences). For instance, customers are known to exhibit a strong preference for the first-come, first-served system (FCFS), which has the least dispersed waiting time distribution among all queueing disciplines.  Second, nontrivial time preferences interact with the effects of other policy variables in a manner that makes it difficult to isolate their effect.  Linear waiting times, which imply risk-neutral time preferences, help isolate the channel of effects orthogonal to those caused by nonlinear time preferences. } This means that, if a buyer waits $t\ge 0$ and receives the good/service, he enjoys the payoff of $v-ct$; his payoff is $-ct$ in case he never receives the good. His outside option is zero.\footnote{Note that a buyer stops incurring the waiting cost once he leaves the queue.  One can think of the cost as the opportunity cost of not exercising the outside option, which would yield the flow value of $c$. An outside option could be a leisure activity that the buyer is forgoing or the next best good or service he has immediate access to.}

The designer's job is to allocate goods to buyers.  In case no item is available when a buyer arrives, the designer may hold the buyer in a (metaphorical) waiting room, called a {\bf queue}, until an item arrives.\footnote{We use the term ``queue'' in a broad sense without any connotation about the service priority rule, such as first-come, first-served. }  
Likewise, if no buyer is available, when an item arrives, the designer may hold it in  inventory until a buyer arrives; however, doing so incurs cost  $d>0$ to the designer per unit of time.  There is also a firm that provides service or good; it earns a fixed profit $\pi>0$ per each item or service rendered to a buyer.  

One can see that this model encompasses two modal scenarios of modern marketplaces. 
\begin{itemize}
    \item {\bf Service scenario:}  Service centers provide perishable goods or services.  Cloud computing, repair and maintenance services, and customer services all fall into this category.  In the context of the service center, the arrival of goods is interpreted as the completion of service for buyers who are already receiving service. Specifically, as long as there is at least one buyer in the system, service is being provided; at the Poisson rate $\mu$, a service completion occurs. If a service completion 'arrives' when there are no buyers (i.e., the system is empty), the service opportunity is wasted.  
 This is a special case, where $d=\infty$ so that only buyers can wait in a queue, and items cannot be stored.

   \item {\bf Goods and dynamic matching scenario:}  Retail platforms, dating apps, public housing authority, human organs transplant organizations, child adoptions and foster care agencies mediate matches between agents/resources on two-sided markets.\footnote{While some of these examples, such as dating apps or retail platforms, incorporate monetary transfers, their role in mediating specific matching and allocation decisions is often limited. For instance, in dating apps like Tinder, pricing is primarily used to discriminate among users for access to a more informed pool (e.g., those who have already 'liked' the user). In contrast, the actual matching follows the platform's NTU-based recommendation algorithm.}  In these environments, entities on both sides can be held in queues if immediate matches are impossible or undesirable. 
\end{itemize}

\paragraph{Mechanisms.}  
A canonical probability space $(\Omega, \mathcal{F}, P)$ captures the primitive arrival processes. A history $\omega \in \Omega$ is a realization of two independent Poisson processes tracking the arrival times of buyers, $\{a_{i}(\omega)\}_{i \in \mathbb{N}}$, and items, $\{g_{j}(\omega)\}_{j \in \mathbb{N}}$, indexed by their arrival order. Let $\{\mathcal{F}_{t}\}_{t \ge 0}$ be the natural filtration generated by these processes. A \textbf{mechanism}, $\phi$, is a non-anticipatory,\footnote{That is, the mapping up to time $t$ must be adapted to the filtration $\mathcal{F}_{t}$.} measurable mapping from histories to outcomes. It specifies which buyers and items are queued, when they are removed, and how they are matched.\footnote{While one could formally specify the outcome space and the induced stochastic processes in greater detail, we omit this formalism here to maintain narrative flow. Furthermore, the methodology of this survey, which relies on a relaxed program, does not require the full machinery of the general mechanism. For readers interested in the complete formal specification of dynamic mechanisms and histories in these contexts, we refer them to the online appendices of \cite{che2021optimal} and \cite{che2025optimal}.}  We also impose an efficiency condition, \textbf{No Allocation Delay (NAD)}, meaning the mechanism never holds both buyers and items simultaneously---an assumption that is without loss for the objectives we consider.

A mechanism $\phi$ induces an outcome $y \in \mathcal{Y}$ at each time $t$, specifying the set of current matches, admissions to the queue, and removals from the queue. An \textbf{outcome process} is then the stochastic process $\{y_t\}_{t \ge 0}$ induced by the mechanism's mapping of histories to outcomes. In this sense, the outcome process is the realization of the mechanism’s decisions over time.  This, in turn, induces a coarser \textbf{queue state process}, $\{\theta_{t}\}_{t \ge 0}$. The state $\theta_{t} \in \Theta := \mathbb{Z}$ represents the number of waiting participants:  $\theta_t>0$ indicates $\theta_t$ buyers in the queue; $\theta_t<0$ indicates $-\theta_t$ items in the inventory; and $\theta_{t}=0$ is the null state where the system is empty. A time $\tau$ is a \textbf{null time} if the queue is empty, $\theta_{\tau}=0$.

Given Markovian arrivals, the system probabilistically restarts after each null time, so the designer repeatedly faces the same problem. It is therefore without loss to require a mechanism to depend only on the history since the last null time.  Formally, let $\omega|_{t}$ be the history following time $t$, with time and participant indices reset. A continuation mechanism $\phi|_{t}$ is the outcome process following time $t$, again with the time and buyer/item indices reset. A mechanism is \textbf{regenerative} if $\phi|_{\tau}(\omega)=\phi(\omega|_{\tau})$ for each null time $\tau$. Restricting attention to regenerative mechanisms is without loss of generality.


If a mechanism is regenerative, its induced queue-state process $\{\theta_{t}\}$ is also regenerative. This class is very general, nesting \textbf{Markov mechanisms}---where decisions depend only on the current state $\theta_{t}$---as a special case. However, the Markov class is often insufficient. For instance, being oblivious to the arrival order, it cannot implement standard rules such as First-Come, First-Served (FCFS). Although arrival orders are payoff-irrelevant in a memoryless process, they can be instrumental for maintaining dynamic incentives, as we will see.

A remarkable feature of a regenerative process $\{\theta_{t}\}$ is that if it is also \textit{positive recurrent}---meaning the expected return time to a null state is finite---it admits a unique stationary distribution $p \in \Delta(\mathbb{Z})$. This distribution describes the system's long-run behavior, as empirical time-averages of queue states converge to $p$ almost surely.\footnote{See \cite{asmussen2003applied} (p. 170, Theorem 1.2) and \cite{THORISSON1992237}; we collect a few relevent results in \Cref{app:regen}.} We therefore focus on a \textbf{Positive-Recurrent Regenerative Mechanism (PRRM)}, which induces such a process. This restriction is without loss, as a non-positive recurrent mechanism would lead to unbounded queues and infinite expected costs, which is never optimal.\footnote{A process that is not positive recurrent is either \textit{transient} (queue length diverges) or \textit{null recurrent} (expected return to the null state is infinite). Both imply unbounded queue growth and infinite costs, making them suboptimal.  A third possibility is that the queue-state process is positive recurrent but never reaches the null state. A stationary distribution is still well defined and unique in this case, and our formulation is valid.}

\paragraph{Incentives.}   The incentive issues concerning buyers depend on the set of instruments employed by the designer, including possible control of entry/exit, information policy, and the service priority.  On both sides (firm and buyers), we require that they at least break even so that they participate in the mechanism. On the buyer side, the incentive issue arises because the platform cannot compel a buyer to enter and/or remain in the queue against their will; therefore, any desired behavior in this regard must be incentivized.  An important issue that arises is the belief a buyer forms about the history leading up to his arrival, particularly the current queue state. We assume that a buyer forms his belief about $\t\in \T$ based on the stationary distribution $\mathbf{p}$ induced by the mechanism. A justification is that each arriving buyer knows that the mechanism has been operating for long enough so that the limit distribution applies.\footnote{This is established formally in \cite{wolff1982poisson}.}  Note this is when the buyer observes nothing other than the fact that he ``arrives.'' If the buyer observes the queue state or some additional signal about it, he will have a refined belief.

While we deal with this issue more precisely as we get to the specific results, here we will simply require that a chosen mechanism be {\bf incentive compatible}: i.e., buyers have incentives to obey the recommendation associated with $\phi$ as a Bayes Nash equilibrium.\footnote{The agents, buyers in our model, face dynamic environment, so refinements capturing sequential rationality may be warranted. We use Bayes-Nash equilibria in the formulation of the problem here to broaden our search for the optimal mechanism. As will become clear, either the equilibrium we study in each specific case satisfies sequential rationality (as in \Cref{sec:complete-info}) or the information policy makes Bayes-Nash equilibria a relevant concept (as in \Cref{sec:che-tercieux}).}  We let $\Phi$ denote the set of all incentive-compatible PRRMs.

\paragraph{Problem Statement.}     
We are interested in the following problem: 
$$\max_{\bp\in \Delta(\Z) }  \a \left[\sum_{k\in \N}p_k( \m v-ck) +\sum_{\ell\in \N} p_{-\ell } \l v \right ] +   (1-\a)\left[\sum_{k\in \N}p_k \m \pi+\sum_{\ell\in \N} p_{-\ell } (\l \pi- \ell d) \right ], \eqno{[\mathcal{P_{\ntu}}]}$$
where $\mathbf{p}=(p_i)_{i\in \Z}$  is the stationary distribution induced by some $\phi\in \Phi$.  
One can interpret the objective as the long-run time average of the weighted sum of buyer welfare and profit, where $\a$  is the weight for the buyer welfare.  I have motivated the relevance of the program $[\mathcal P_{\ntu}]$ by the use of PRRM.  Beyond PRRMs, this program is well-defined whenever the mechanism the designer employs induces a process on the queue states $\T$ that admits a well-defined stationary distribution.


\paragraph{Benchmark: Relaxed problem.}   Before jumping into specific results, it is useful to begin with the first-best benchmark, which ignores the incentive constraints, except for participation constraints. Namely, consider 

$$\max_{\bp\in \Delta(\Z) }  \a \left[\sum_{k\in \N}p_k( \m v-ck) +\sum_{\ell\in \N} p_{-\ell } \l v \right ] +   (1-\a)\left[\sum_{k\in \N}p_k \m \pi+\sum_{\ell\in \N} p_{-\ell } (\l \pi- \ell d) \right ], \eqno{[\mathcal{P}_{\ntu}']}$$
subject to  
\begin{align}
 &\sum_{k\in \N}p_k( \m v-ck) +\sum_{\ell\in \N} p_{-\ell } \l v\ge 0; \tag{${IR_B}$}  \label{IRB}\\
 &\sum_{k\in \N}p_k \m \pi+\sum_{\ell\in \N} p_{-\ell } (\l \pi- \ell d)\ge 0; \tag{${IR_S}$}  \label{IRS} \\
 & \l p_k \ge \m p_{k+1}, \, \forall k \in \N\cup \{0\}; \tag{${B_k}$}  \label{BB} \\
  & \m p_{-\ell} \ge \l p_{-\ell-1}, \, \forall \ell \in \N\cup \{0\}.  \tag{${B_{-\ell}}$}  \label{BS} 
\end{align}

The first two constraints, \cref{IRB} and \cref{IRS} are the break-even, or participation, constraints for the firm and buyers.  The two constraints, \cref{BB} and \cref{BS}, are necessary for the stationarity of $\bp$.  \Cref{fig:balance} helps explain \cref{BB}.  Fix any $k\ge 0$, consider the probability that the queue state transitions from $\T_{\le  k}:=\{j\in \Z: j\le k\}$ (to the left of the dotted line in \Cref{fig:balance}) to  $\T_{>k}:=\T\setminus {\T_{\le k}}$ (to the right of the dotted line)  for a brief instant $dt>0$.  This probability is at most $p_k \l dt+o(dt)$---the probability that there are $k$ buyers and one more buyer arrives during $dt>0$---``at most,'' since the new arrival may or may not join the queue.  Likewise, for a brief instant $dt>0$, the probability that the queue state transitions from $\T_{>k}$ to $\T_{\le  k}$ is at least $p_{k+1} \mu dt+o(dt)$---the probability that there are $k+1$ buyers and an item arrives during $dt>0$, in which case a buyer leaves with a good. This is a lower bound, given our assumption of NAD.  The latter can't exceed the former at a stationary distribution, which explains \cref{BB}.

\begin{figure}[h!]
    \centering
\begin{tikzpicture}[
    >=stealth',
    node distance=3.5cm,
    on grid,
    auto,
    state/.style={
        circle,
        draw,
        minimum size=1.5cm
    }
]

\node[state] (k-1)              {$k-1$};
\node[state] (k)   [right=of k-1]   {$k$};
\node[state] (k+1) [right=of k]     {$k+1$};
\node[state] (k+2) [right=of k+1]   {$k+2$};

\node (dots1) [left=of k-1, xshift=1cm] {$\dots$};
\node (dots2) [right=of k+2, xshift=-1cm] {$\dots$};

\path[->, thick, shorten >=1pt]
    (k-1) edge [bend left]  node[above] {} (k)
       (k)   edge [bend left]  node[below] {} (k-1)
    (k+1)   edge [bend left]  node[below, text=red] {$\mu$} (k)
    (k)   edge [bend left]  node[above, text=red] {$\lambda x_k$} (k+1)
    (k+1) edge [bend left]  node[below] {} (k)
    (k+1) edge [bend left]  node[above] {} (k+2)
    (k+2) edge [bend left]  node[below] {} (k+1)

    (k-1) edge [loop above] node {} ()
    (k)   edge [loop above] node {} ()
    (k+1) edge [loop above] node {} ()
    (k+2) edge [loop above] node {} ();

\draw[->, thick] (k-1) edge[bend left=45] (dots1);
\draw[->, thick] (dots1) edge[bend left=45] (k-1);
\draw[->, thick] (k+2) edge[bend left=45] (dots2);
\draw[->, thick] (dots2) edge[bend left=45] (k+2);

\draw[dotted, thick] ([xshift=1cm]k.east) -- ([xshift=1cm, yshift=2cm]k.east);
\draw[dotted, thick] ([xshift=1cm]k.east) -- ([xshift=1cm, yshift=-2cm]k.east);


\end{tikzpicture}
    \caption{Balance conditions for a stationary distribution.}
    \label{fig:balance}
\end{figure}
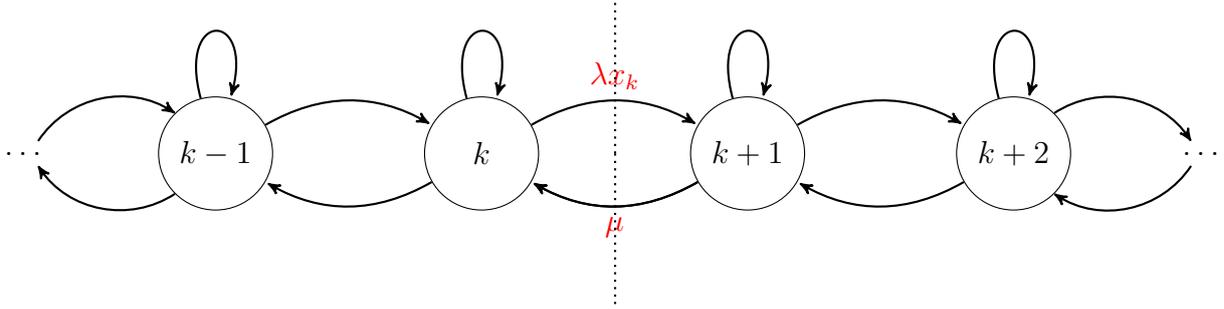
An analogous condition explains \cref{BS}. Any $\bp$ satisfying \cref{BB} and \cref{BS} can be implemented as a stationary distribution, by probabilistically regulating the entry of buyers or goods into the queue.  Specifically, given   \cref{BB} for $k\in \N$ say, one can choose a suitably-calibrated admission probability $x_k\in (0,1]$, so that the balance condition holds with equality.\footnote{Strictly speaking, a stationary distribution must satisfy a balance condition between any $ \T'\subset  \T$ and $\T\setminus\T'$. However, this is implied by the conditions required here.}

Note that we have reformulated the dynamic mechanism design problem into a static linear programming (LP) problem in which the designer directly chooses a stationary distribution $\bp$. This methodology holds promise for making a dynamic problem tractable and is the primary methodological focus of the current survey.  We can characterize the first-best as involving a cutoff structure:

\begin{theorem}  \label{thm:cutoff}   There exist $K^*, L^*\in \{0\}\cup \N\cup\{\infty\}$ such that the optimal solution $\bp$ to $[\mathcal P_{\ntu}']$ has support $\{-L^*, ..., K^*\}$  and satisfies $(B_k)$ with equality for all $k\in \N$ with $k<K-1$ and  $(B_{-\ell})$ with equality for all $\ell\in \N$ with $L^*-1$.  For $d$ sufficiently large, $L^*=0$, and for $c$ sufficiently large, $K^*=0$.
\end{theorem}

\noindent {\it Proof Sketch.} The interval structure of the support of $\bp$ is obvious. The support contains $0$ since we assume no allocation delays, which makes $\t=0$ positive recurrent.  To see that \cref{BB} and \cref{BS} are binding for $k<K^*-1$ and $\ell<L^*-1$, write the objective as 
$$\sum_k f(k)p_k +\sum_{\ell}g(\ell)p_{-\ell},$$
where $f(k):= (\a+\nu)(\m v-ck) +(1-\a+\g)(\m \pi)$  and  $g(\ell
):= (\a+\nu)\l v +(1-\a+\g)(\l\pi -\ell d)$, and $\nu$ and $\g$ are (nonnegative) Lagrangian multipliers respectively for \cref{IRB} and \cref{IRS}.  Observe that $f(\cdot)$ is strictly decreasing.  Hence, if \cref{BB} is slack for some $k<K^*-1$, then one can raise $p_k$ slightly and lower $p_j$ for $j \in \{k+1, ..., K^*\}$ without violating any constraints. This modification strictly increases the objective.  An analogous argument works for $j=-\ell$, with $\ell<L^*-1$.  See  \cite{che2021optimal} for details.  \qed
\bigskip

A binding \cref{BB} means that a buyer arriving at a queue of $k$ buyers must enter the queue with probability one.  Therefore, the optimum in \Cref{thm:cutoff} can be implemented by a {\bf cutoff policy} which admits buyers (goods) into the queue only up to some critical number $K^*-1$ (resp. $L^*-1$) and denies entry once $K^*$ (resp. $L^*$) is reached.\footnote{Partial rationing may be used for $\t=K-1$ (resp  $\ell=L^*-1$) when \cref{IRB} (resp. \cref{IRS})  is binding; otherwise, there is no rationing.} The buyers who join the queue are asked to stay until they receive items/services. Of course, whether and how the buyers may be incentivized to follow these instructions are the key questions that will be discussed below. 

Intuitively, a cutoff policy is optimal since, while storing an additional buyer (resp. good) has the option value against the future arrival of goods (resp. buyers), the value of this option diminishes as we store more and more buyers (resp. goods). Consequently, if it pays to admit a new buyer when there are $k$ buyers in the queue, it does so when there are $j<k$ buyers in the queue.

\subsection{Complete Information  Analyses}\label{sec:complete-info}

\cite{naor1969regulation} pioneered the {\it rational queueing theory}, which studies the incentives facing the agents (buyers in our context) in a queueing environment; see \cite{hassin-haviv2003} for a survey. Much of this literature considers a setting in which buyers have complete information about the queue state: namely, when they arrive at the queue, they observe the number of buyers already in the queue.  This assumption is less compelling in many modern settings, such as service call centers, even kidney exchanges, and digital platforms, in which queue information can be, and often is, withheld from the agents.  However, it remains relevant in specific physical queue settings, such as grocery checkouts and emergency room waiting areas.  Below, we assume $d=\infty$, as is realistic in the service queue setting, and considered by all authors.  At the same time, the results below will remain qualitatively valid even when $d<\infty$.

\subsubsection{Excessive queueing under FCFS.}

\cite{naor1969regulation} considers the welfare objective with $\alpha=1$ and studies the queueing incentive of the buyers who observe the queue state. The queue discipline is FCFS.

To begin, suppose a buyer is the first to arrive at an empty queue.  He can start receiving  service immediately, and it takes  $1/\m$ on average for service to be completed (since service completion occurs at the Poisson rate of $\m$). The buyer will join the queue if and only if $v>c/{\m}$.

Suppose next a buyer arrives with $k-1$ buyers already in the queue.  Under FCFS, his service begins only after all $k-1$ buyers ahead of him are served, each taking time distributed exponentially with mean $1/\m$. So, his total waiting time (including his service time) is Erlang-distributed with mean $k/\m$.\footnote{A Erlang distribution with $(k,\m)$ is the sum of $k$ independent exponentially distributed random variables with mean $1/\m$: see \url{https://en.wikipedia.org/wiki/Erlang_distribution}.} Hence, he will wish to join the queue if and only if $v\ge ck/{\m}$.  Note that once a buyer joins the queue, he can only move up in the priority order, so his residual mean wait time can only fall, meaning he will want to remain in the queue until he is served.  

It follows that the equilibrium will be of the cutoff structure specified in \Cref{thm:cutoff}, with the maximal queue length given by the marginal buyer (almost) indifferent to joining the queue, or 
$$K_{\fcfs}:=\left\lfloor \m \frac{v}{c} \right\rfloor=\max\left\{k: v- c\frac{k}{\m}\ge 0\right\}.$$
The key question is: {\it how does $K_{\fcfs}$ compare with the optimal cap $K^*$? }
To answer this question, recall that we evaluate the welfare in the stationary distribution. Suppose that agents queue up  to some maximal $K$. Requiring the balance condition  \cref{BB} for $k=1,..., K$ with equality, we observe
$$p_k^K=\frac{\l}{\m}p_{k-1}^K= \left(\frac{\l}{\m}\right)^2p_{k-2} = ...=\left(\frac{\l}{\m}\right)^k p_0^K=\rho^k p_0^K, $$
where $\r:=\l/\m$ is the ``load  factor.'' Using $\sum_{k=0}^K p_k^K=1$, we obtain 
$$p_k^K= \frac{\rho^k}{1+\rho +...+  \rho^K}, \, \forall k=0,..., K.$$
It is important to observe that $p_k^K$  is decreasing in $K$ for each $k\le K$; namely, if more buyers are willing to queue up, the system spends less time on a lower queue state.  This reflects the negative externality conferred by a marginal buyer: {\it a buyer joining the queue increases the queue length experienced by future buyers and their waiting times}. The presence of such a congestion externality leads to the following conclusion:

\begin{theorem} Assume  $d=\infty$ and  $\a=1$.  We have $K^*\le K_{\fcfs}$, where the inequality is often strict, and the queue length under FCFS is larger in first-order stochastic dominance than that under optimum. 
\end{theorem}
\begin{proof}  The second part follows  from the first part, since, given $K^*\le K_{\fcfs}$, 
$$\sum_{j=0}^k p_j^{K_{\fcfs}}\le\sum_{j=0}^k p_j^{K^*}, \, \forall k.$$
To prove the first part, recall that we can write the objective as $\sum_{k=0}^{K} p_k^K f(k),$
where $f(k)= \m v - k c$. (Recall $\a=1$ and $L^*=0$ since $d=\infty$.) Now observe that for all $k\le K_{\fcfs}$,  $f(k)\ge 0$.  Consider any $K>K_{\fcfs}$. Since $K_{\fcfs}:=\left\lfloor \m \frac{v}{c} \right\rfloor$,  $f(j)<0$ for each $j=K_{\fcfs}+1, ...,K$,  and $p_k^K< p_k^{K_{\fcfs}}$ for each $k\le K_{\fcfs}$.  It thus follows that $\sum_{k=1}^{K} p_k^K f(k)< \sum_{k=1}^{K_{\fcfs}} p_k^{K_{\fcfs}}f(k)$, whenever $K>K_{\fcfs}$. We thus conclude that $K^*\le K_{\fcfs}$.
\end{proof}

The intuition can be seen clearly in the standard textbook analysis of a market.  \Cref{fig:surplus_comparison} compares the surplus under FCFS and under the optimum (with $\a=1$) identified in \Cref{thm:cutoff}, when  $v=5, c=1$, and $\rho=\lambda=\mu=1$.  In each case, the horizontal dotted line depicts the value realized $\m v=5$,  and the 45 degrees line (dotted) depicts the social costs associated with waiting, for each state $k$.   

FCFS is depicted in Panel (a), in which $K_{\fcfs}=v/c=5$, and the social surplus $\m v- ck =5-k$ at each state $k$ is simply the difference between the value and the cost curve. Based on standard textbook reasoning, it may be tempting to conclude that FCFS is socially optimal, as it exhausts all possible surplus-generating opportunities by selecting $K$ to be the highest $k$ with nonnegative surplus.  This logic, however, overlooks the ``intensive margin'': as noted earlier, with a higher $K$, each lower inframarginal state $k<K$, where the surplus is higher, becomes less likely to occur.   By reducing the queue cap $K$, welfare may increase due to improvements on the intensive margin.

In the example, with $K_{\fcfs}=5$,
$p_1^{K_{\fcfs}}=...=p_5^{K_{\fcfs}}=1/6$, and the welfare is:  
$\sum_k p_k^K (\m v- ck) =(4+3+2+1)/6=5/3$.  By contrast, under the optimal outcome with $K^*=2$, the likelihood of infra-marginal state increases to $p_1^{K^*}=p_2^{K^*}=1/3$, and the welfare rises to $(4+3)/3=7/3$. The increased likelihood of states $k=1$ and $k=2$ under $K^*$ is represented in panel (b) as thicker lines than in (a).

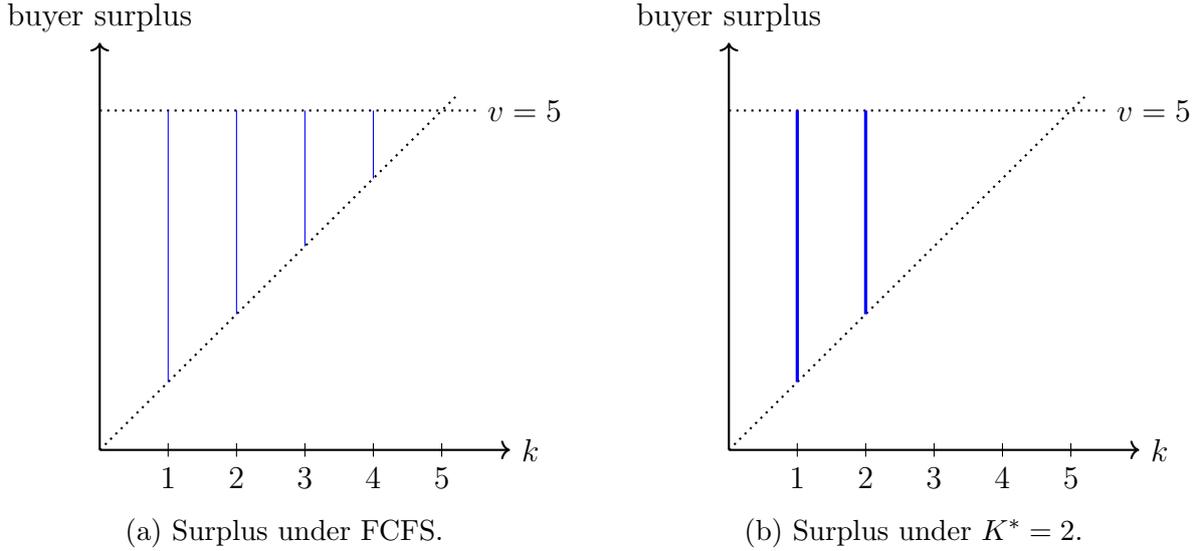
\begin{figure}[h!]
\centering

\begin{subfigure}{0.48\textwidth}
    \centering
    \begin{tikzpicture}[scale=0.9]
        \draw[->, thick] (0,0) -- (6,0) node[right] {$k$};
        \draw[->, thick] (0,0) -- (0,6) node[above] {buyer surplus};

        \foreach \k in {1, 2, 3, 4, 5} {
            \draw (\k, 0.1) -- (\k, -0.1) node[below] {$\k$};
        }
        
        \draw[dotted, thick] (0,0) -- (5.2,5.2); 
        \draw[dotted, thick] (0,5) -- (5.5,5) node[right] {$v=5$};

        \foreach \k in {1, 2, 3, 4} {
            \draw[blue] (\k, \k) -- (\k, 5);
        }
    \end{tikzpicture}
    \caption{Surplus under FCFS.}
    \label{fig:panel_a}
\end{subfigure}
\hfill 
\begin{subfigure}{0.48\textwidth}
    \centering
    \begin{tikzpicture}[scale=0.9]
        \draw[->, thick] (0,0) -- (6,0) node[right] {$k$};
        \draw[->, thick] (0,0) -- (0,6) node[above] {buyer surplus};

        \foreach \k in {1, 2, 3, 4, 5} {
            \draw (\k, 0.1) -- (\k, -0.1) node[below] {$\k$};
        }

        \draw[dotted, thick] (0,0) -- (5.2,5.2); 
        \draw[dotted, thick] (0,5) -- (5.5,5) node[right] {$v=5$};
        
        \foreach \k in {1, 2} {
            \draw[blue, very thick] (\k, \k) -- (\k, 5);
        }
    \end{tikzpicture}
    \caption{Surplus under $K^*=2$.}
    \label{fig:panel_b}
\end{subfigure}

\caption{Comparison of buyer surplus.}
\label{fig:surplus_comparison}

\vspace{1em} 
\small{Note: $v=5, c=1$, and $\rho=\lambda=\mu=1$.}
\end{figure}

\subsubsection{LCFS to the rescue?}

\cite{hassin1985} made an interesting observation that the Last-Come, First-Served (LCFS) restores the welfare optimal outcome identified in \Cref{thm:cutoff}. The idea is as follows.  When LCFS is used, there is always a strict incentive for an arriving buyer to join the queue, as long as he has the option to leave the queue at a later time, an option we assume to exist here.  Hence, each arriving buyer always joins the queue.  The buyer in the queue who arrived earliest, let's call him {\it the first incumbent},  is last to be served, so he is most willing to leave the queue if it becomes too long. In other words, the maximal queue length is effectively chosen in LCFS by the first incumbent's decision to leave the queue.  The key insight boils down to the following difference between FCFS and LCFS.  In the FCFS,  when a buyer joins the queue when $k=K_{\fcfs}-1$, part of the social cost of that decision is externalized to the buyers who will arrive later and have to wait longer;\footnote{Such externalities exist even if no new buyer joins the queue when $k= K_{\fcfs}$, since the queue length eventually becomes strictly below $K_{\fcfs}$ and new buyers arrive and have to wait longer because the marginal buyer decided to join the queue.} however, when the first incumbent decides to stay when the new entry triggers the queue length to increase from $K^*$ to $K^*+1$, nobody else bears that social cost, except for the first incumbent.  This logic suggests that the first incumbent leaves the queue if and only if a new entry causes the queue length to reach $K^*+1$.

While this logic is compelling, \cite{hassin1985} offers no formal analysis or proof. I provide a simple analysis here, whose proof, along with most of the proofs, appear in the Appendix.

\begin{theorem} \label{thm:hassin} Assume  $d=\infty$ and  $\a=1$.
LCFS implements the optimal solution of [$\mathcal P_{\ntu}'$] with $K_{\lcfs}=K^*$.
\end{theorem}

While insightful and valuable in isolating the source of inefficiency under FCFS, the practicality of LCFS is dubious. First, the system may be gamed; once they lose their last arrival status, the agents may leave the queue and reenter it, thereby regaining that status and the priority that comes with it.  While the designer may manage that problem by prohibiting reentry, this requires the designer to verify agents'  identities, which the designer may find costly or undesirable.  Second, the analysis and results implicitly assume ``preemption," meaning that a server may switch from an existing buyer to a new one when the latter arrives, possibly while serving the former.  While this is payoff-irrelevant given the memoryless property of the exponential service process, in practice, productivity losses may occur due to service interruptions.  While FCFS does not suffer from such losses, LCFS  does. Last but not least, the psychology of queueing is clear that individuals experience strong resentment when the arrival order is not respected, let alone reversed; see \cite{larson1987}.  

\subsubsection{Insufficient queueing under FCFS} \cite{leshnoAER} considers a situation in which the assigned goods are so over-demanded that agents' waiting costs do not constitute social costs (provided that, once a good becomes available, it is assigned to somebody without delay). Namely, in his model, the total number of people who must wait remains independent of who is served. The goal is instead to keep them waiting until goods with high match values become available, no matter how long it may take.  

He focuses on the priority rule as an instrument to maximize the incentive to wait for the matched item.  The main tradeoff across different priority rules can be captured isomorphically within our framework by assuming $d=\infty$ and $\a=0$. Just like the planner in \cite{leshnoAER}, our designer with $\a=0$ wishes to maximize buyers' incentives for waiting, to minimize the risk of ``buyer stockout'' or to maximize the option value of future matching when goods arrive.

In this case, either  $K^*=\infty$ (if $\m>\l$), or \cref{IRB} is binding at $K^*$. That is, if  $K^*<\infty$, then
\begin{align} \label{eq:1}
    \sum_{k=1}^{K^*} p_k^{K^*} (\m v - kc)\approx0,
\end{align}
where we use $\approx$ to mean that $K^*$ is the smallest $K$ among those that make the LHS nonpositive.\footnote{More precisely, the rationing of entry in state $k=K^*-1$ is calibrated to satisfy the equality exactly. The chosen $K^*$ with the rationing will be precisely the smallest $K$ among those that make the LHS nonpositive.} The condition \cref{eq:1} means that 
$$K^*\ge\m \frac{v}{c} \ge  K_{\fcfs},$$
with the first inequality being strict whenever $K^*>1$, in which case FCFS is suboptimal.  Unlike \cite{naor1969regulation}, the problem now is that buyers queue too little, rather than too much.  The difference is due to the objective function: $\a=1$ versus $\a=0$.

\cite{leshnoAER}'s main result is that Service in Random Order (SIRO) improves on FCFS.  SIRO assigns the good uniformly at random amongst the buyers in the queue.  Given this rule, a shorter queue is still more profitable to join than a longer queue, so the entry decision still has a cutoff structure: there exists some $K_{\siro}$ such that a buyer enters the queue if and only if $k<K_{\siro}$  and those who enter the queue stay until they receive the goods.

\begin{theorem} \label{thm:leshno} Assume $d=\infty$ and $\a=0$.  Then, $K_{\fcfs}\le K_{\siro}\le K^*$, and SIRO attains a weakly higher (resp. lower) value of the objective for the designer than does FCFS (resp. the optimum).
\end{theorem}

The intuition behind \Cref{thm:leshno} is made clear via the following example (adapted from \cite{che2021optimal}).  
Assume $\m=\l=1$ and $  \frac{3 }{2 } c\le v<2c$. The former inequality implies that $K^*=2$, meaning the optimal cutoff policy accommodates up to two buyers in the queue,\footnote{Little's law can be used to show that the ex ante expected wait time conditional on there being a buyer for queue with $K=2$ equals: $\frac{p_1^2+2p_2^2}{\l(p_0^2+p_1^2)}=3/2$, since $p_0^2=p_1^2=p_2^2=1/3$ given $\m=\l=1$.}  and the implication of the latter condition will become clear. 

\Cref{fig:leshno} plots the expected wait costs for a buyer who arrives at an empty queue ($k=0$) and a queue with one buyer ($k=1$) under FCFS.  Under FCFS, a buyer will expect to wait  $1/\m=1$  unit of time for  the good if $k=0$ and $2/\m=2$ units of time if $k=1$.  The associated waiting costs are $c$ and $2c$, respectively. If  $v<2c$, as described in the figure, then a second buyer will refuse to join the queue, so FCFS supports at most one buyer in the queue.

\begin{figure}[h!]
\centering

\begin{subfigure}[b]{0.45\textwidth}
\centering
\begin{tikzpicture}[scale=2,font=\small]

\draw (-0.2,0) -- (3.5,0) node[right] {$k$};
\draw (0,-0.2) -- (0,2.5) node[above] {Waiting costs};

\filldraw[black] (2,2) circle (0.3mm);
\filldraw[black] (1,1) circle (0.3mm);
\filldraw[red] (2,1.635) circle (0.3mm);
\filldraw[red] (1,1.32) circle (0.3mm);

\node[left] at (0,2) {$2c$};
\node[left] at (0,1) {$c$};
\node[left] at (0,1.75) {$v$};
\node[below] at (1,0) {$k=0$};
\node[below] at (2,0) {$k=1$};

\draw[red,->,thick] (2,2) -- (2,1.64);
\draw[red,->,thick] (1,1) -- (1,1.31);

\draw (0,1.75) -- (3.3,1.75);

\draw[dotted] (1,0) -- (1,2.3);
\draw[dotted] (2,0) -- (2,2.3);
\draw[dotted] (0,1) -- (1,1);
\draw[dotted] (0,2) -- (2,2);

\end{tikzpicture}
\caption{SIRO supports 2 agents in the queue; FCFS can't.}
\end{subfigure}
\hfill
\begin{subfigure}[b]{0.45\textwidth}
\centering
\begin{tikzpicture}[scale=2,font=\small]

\draw (-0.2,0) -- (3.5,0) node[right] {$k$};
\draw (0,-0.2) -- (0,2.5) node[above] {Waiting costs};

\filldraw[black] (2,2) circle (0.3mm);
\filldraw[black] (1,1) circle (0.3mm);
\filldraw[red] (2,1.64) circle (0.3mm);
\filldraw[red] (1,1.32) circle (0.3mm);
\node[left] at (0,2) {$2c$};
\node[left] at (0,1) {$c$};
\node[left] at (0,1.5) {$v$};
\node[below] at (1,0) {$k=0$};
\node[below] at (2,0) {$k=1$};

\draw[red,->,thick] (2,2) -- (2,1.64);
\draw[red,->,thick] (1,1) -- (1,1.31);
\draw (0,1.5) -- (3.3,1.5);

\draw[dotted] (1,0) -- (1,2.3);
\draw[dotted] (2,0) -- (2,2.3);
\draw[dotted] (0,1) -- (1,1);
\draw[dotted] (0,2) -- (2,2);

\end{tikzpicture}
\caption{SIRO can't support 2 agents.}
\end{subfigure}

\caption{Comparison of queueing rules under full information} 
\footnotesize
Notes: The black and red dots represent the expected waiting cost under FCFS  and SIRO, respectively.
\label{fig:leshno}
\end{figure}
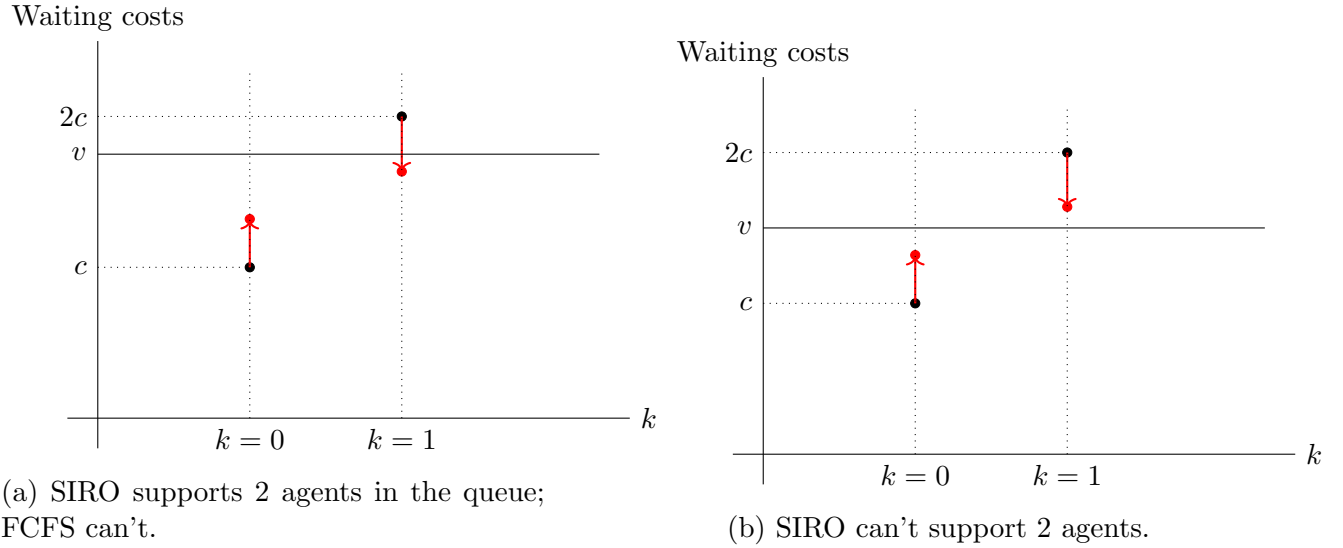

Note that a buyer arriving in state $k=0$ enjoys a strictly positive surplus.  If one can transfer this ``slack'' incentive to a buyer arriving in state $k=1$, the latter may be incentivized to join the queue. This is precisely what SIRO does by ``flattening'' the waiting cost curve.   By randomizing service priority, SIRO increases the priority of a buyer arriving in state $k=1$ at the expense of the buyer arriving in state $k=0$, compared to FCFS. 

Observe that the waiting cost curve is not completely flat even under SIRO.\footnote{The mean waiting times are   $4/3$ if $k=0$ and $5/3$ if $k=1$; see the equations in the proof of \Cref{thm:leshno}.}  If $v$ is sufficiently high (panel (a)), SIRO can restore the optimum.  However, if $v$ is lower (panel (b)), in particular if $v<\frac{5}{3}c$,  SIRO can't support two buyers in the queue; see panel (b) of \Cref{fig:leshno}. 
 
 It is possible to completely flatten the waiting cost curve by transferring priority to the second arrival even further---under the rule called the Load Independent Expected Waiting  (LIEW) by  \cite{leshno2019dynamic}).  Such a rule maximizes the queueing incentive by equalizing the wait times for a buyer joining in state $k=0$ versus joining in state $k=1$.  In the example, LIEW supports two agents in the queue, provided that $v$ is no less than $1.5c$, the conditional mean waiting cost for a queue with $K=2$. That is, LIEW can attain the optimal solution of $[\mathcal P_{\ntu}']$ (with $\a=0$).
 
 In summary, FCFS is suboptimal and outperformed by rules such as SIRO and LIEW when the queue is fully observable.
However, if the designer can control information,  the queueing rule becomes irrelevant for the incentive to join the queue.  Suppose the designer informs buyers only whether $k\in \{0,1\}$ (``recommend to join'')  or not.  Then, a buyer can't distinguish between $k=0$ and $k=1$ upon his arrival; he then forms the identical belief between $k=0$ and $k=1$, thus facing the expected waiting time of $1.5$, just like LIEW. In other words, the information policy can replace the LIEW.  Under this (optimal) information policy, the {\it mean} wait time---and hence the incentive to join the queue---is the same across {\it all} queueing rules, so they cannot be differentiated on this account.  

Meanwhile, LIEW may create pathological incentives for buyers once they join the queue. Just like LCFS, after joining the queue, a buyer realizes that his remaining expected wait time increases as time passes.  This is because LIEW, to equalize the wait time between a buyer arriving in state $k=0$ and a buyer arriving in state $k=1$, must worsen the former buyer's priority when a second buyer arrives, putting him in a position akin to one who arrives in $k=1$ under FCFS.  This means that the buyer will be tempted to exit the queue once a second buyer joins and takes away his priority.  Consequently, LIEW (and any other rules) can't implement the optimal outcome under complete information.   

\subsection{Information Design and the Optimality of FCFS}\label{sec:che-tercieux}

While the complete information assumption is realistic for some classic contexts, modern businesses, particularly platforms, do have substantial control over what information buyers can or can't have.  Call centers, a quintessential example of queue systems, often leave customers without precise information about how long they have to wait. Public housing systems, a motivating example in \cite{leshnoAER}, often keep the eligible recipients in the dark about their spots in the waitlists.
Digital platforms, which engage customers through mobile or web interfaces, have even more control over information.  Ride-hailing platforms, such as Uber or Lyft, inform customers about the locations of matched drivers but do not reveal other potential matches or allow them to choose from these options. 

Information design confers the designer additional power to control buyers' incentives.  Recall that under complete information, FCFS offers excessive incentives for queueing from the consumer surplus standpoint $(\a=1)$ but insufficient incentives from the producer surplus standpoint $(\a=0)$.  As we already hinted in the last example, information design can help overcome the latter problem. But without additional tools, information design  alone cannot implement the optimal solution $[\mathcal P_{\ntu}']$  incentive-compatibly.

\cite{lingenbrink2019optimal} and \cite{Anunrojwong2020} studied information design by the designer who seeks to maximize consumer welfare (i.e.,  $\a=1$) under FCFS. They showed that the optimal information involves two signals; either the current queue  state  is ``$k<K$'' or ``$k=K$,'' for some $K\in \N$, with the former signal intended for encouraging buyers to join the queue and the latter for discouraging buyers from doing so.  This result discovers the well-known ``folk'' theorem in economics:  {\it a mechanism designer should reveal no more than the actions that she recommends to an agent,} the simple logic being that any distinct signals leading to the same action can be pooled into one without violating any incentive constraints.\footnote{Since the desired action is incentive compatible under each signal, it must be on average over the multiple signals, and hence under the pooled signal.}  Without admissions control, however, this policy itself cannot keep buyers from joining the queue at the signal ``$k=K$'' unless $K\ge K_{\fcfs}$.  \cite{Anunrojwong2020} shows that a cap $\hat K< K_{\fcfs}$ (where $\hat K$ may possibly equal $K^*$) can be implemented if there is an additional type of buyers who incur no wait costs (or have a very low outside option) and are willing to join the queue, regardless of the current queue length.  Then, the signal ``$k\ge \hat K$'' would mean that there are many more than $K_{\fcfs}$ buyers already in the queue, making it incentive-compatible for a buyer with cost $c>0$ not to join when learning ``$k\ge \hat K$.''

This problem does not arise when the designer can control {\it admissions} into a queue.   Formally, the designer can control admission by refusing to serve buyers whom the designer wishes to keep out.  In practice, a call center discouraging entry (accompanied by a repetitive soundtrack) is often effective for this purpose.  \cite{che2021optimal} consider the designer who controls three sets of instruments: (i) {\it admission control} (buyers can be denied entry into a queue or even removed from it); (ii) {\it service priority allocation} (or queue disciplines); and (iii) {\it information design}.  They consider buyers' incentives not only to  join the queue  but also to stay in the queue, when recommended by the mechanism.  The literature has largely ignored this latter incentive issue.\footnote{The exception is \cite{hassin1985}, where the first incumbent's decision to leave the queue plays a crucial role in sustaining welfare maximum.} \cite{che2021optimal} obtain the following result:  

\begin{theorem} \label{thm:che-tercieux} The optimal solution of [$\mathcal P_{\ntu}'$] is implementable by FCFS under the optimal information.  
\end{theorem}

\Cref{thm:che-tercieux} states that the optimal solution to the relaxed program [$\mathcal P_{\ntu}'$] which ignores the incentives problems can be implemented by an admissions control which keeps the buyers from the queue whenever $K^*$ is reached (with possible rationing at $k=K^*-1$), an information policy which imply issues recommendation to join the queue if $k<K^*$ (subject to possible rationing at $k=K^*-1$), and a priority rule of FCFS.  

For a sufficiently large $\a\in [0,1]$, the problem is the excessive queueing a la \cite{naor1969regulation}; this problem is solved simply by keeping buyers from joining to queue beyond $K^*$.\footnote{Equivalently, as noted below, the optimum $K^*$ can be also implemented by a lottery of the real good (with value $v$) and a null good, calibrated to induce $K^*$.  If a buyer rejects a null good, he loses his priority.} The difficult situation is when $\a$ is so low that buyer have insufficient incentives for queueing, as was seen earlier. We already saw that, given \cref{IRB}, the optimal information policy provides buyers with sufficient incentives to join the queue, regardless of the priority rule.  However, it is not easy to provide them with incentives to stay in the queue once they join, which is what the optimal cutoff policy calls for.

To illustrate, consider the earlier example in which $d=\infty, v=1.5c$, and $\m=\l=1$. We note that $K^*=2$, at which \cref{IRB} is binding. Under optimal information, any arriving buyer invited to the queue expects to wait for 1.5 units of time, assuming he will never leave the queue.  Given \cref{IRB}, this wait time incentivizes buyers to join the queue as long as $k<K^*=2$ under any of the priority rules; see \Cref{app:che-tercieux} for the precise argument.

However, the incentive to stay in the queue once a buyer joins may or may not hold, depending on the priority rule.  \Cref{fig:waitingtime}  plots the mean {\it residual} wait times for a buyer who has spent time $t\ge 0$ on the queue. We consider five standard rules: FCFS, SIRO, LIEW, LCFS, and LCFS-PR, where LCFS-PR is the LCFS with ``preemption,'' namely, a rule in which an old agent leaves when a new agent enters the queue.  
\begin{figure}[h]
\caption{Expected wait times under alternative queueing rules.}
\label{fig:waitingtime}\centering
\includegraphics[height=2.5in]{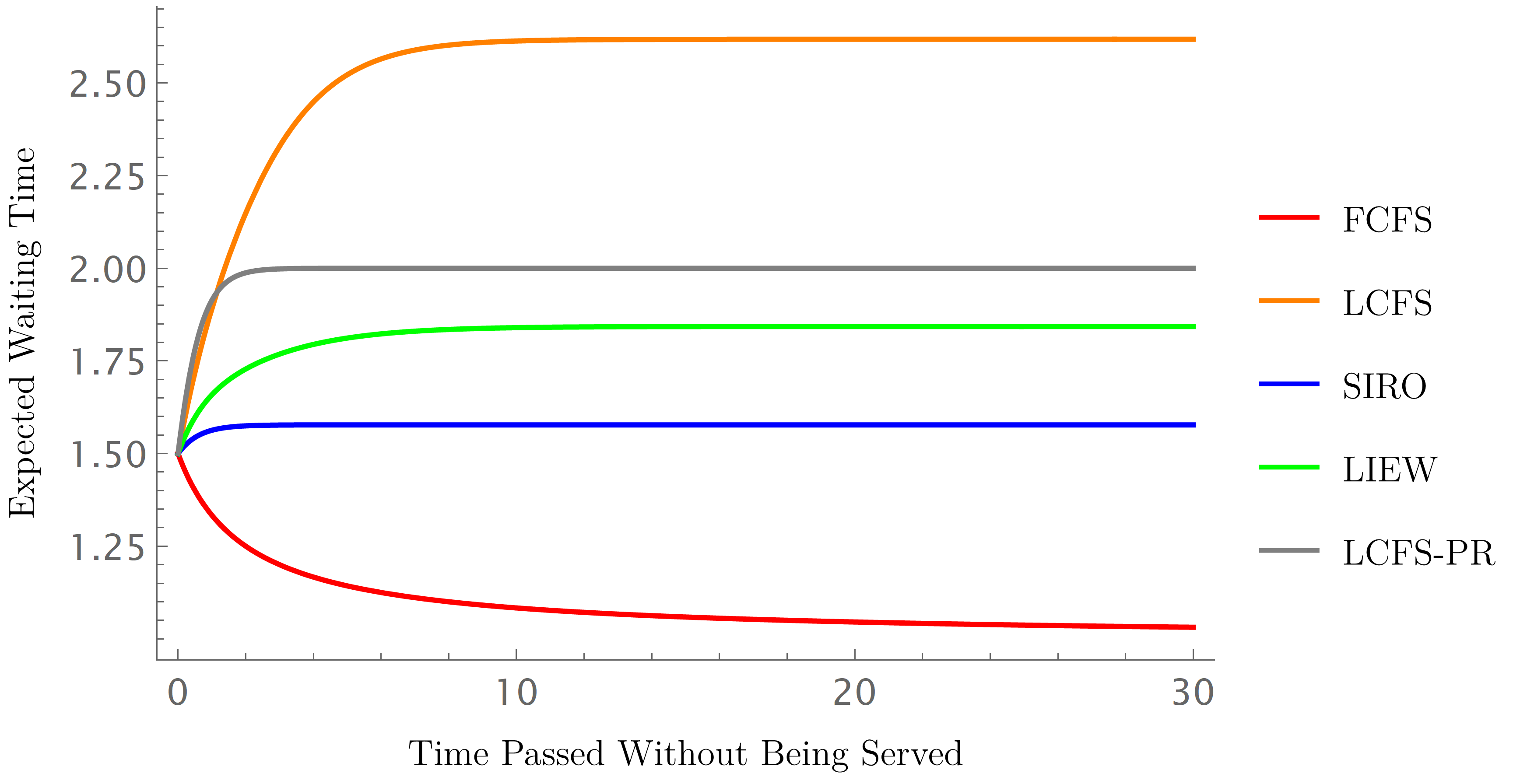}\newline
\end{figure}

Note first that the mean wait time is $1.5$ at $t=0$, as observed below, for all priority rules. Since $v=1.5 c$, buyers are thus indifferent to joining the queue (assuming they will stay until they are served). However, once they have joined, the mean residual wait times diverge under alternative rules as one's time on the queue $t$ rises: it {\it decreases} under FCFS but {\it increases} under all other queueing rules.  Hence, the agents will have incentives to stay until they are served under FCFS, but they will abandon the queue unless served immediately under every other rule. This difference means that the FCFS outperforms the other queueing rules.\footnote{The figure implies that under each of the other queueing rules, it cannot be an equilibrium that agents will join the queue up to two and stay until they are served.  It is unclear what will happen in equilibrium; they may randomize in leaving the queue, and/or they may excessively enter with a plan to leave soon after.  Regardless, \cite{che2021optimal} shows the designer will be worse off than under FCFS for some values of $\l$.}    

Under FCFS, a buyer's service priority increases over time as earlier arrivals leave the queue. Hence, conditional on the initial queue length, one expects to wait less over time, so the mean residual wait time decreases gradually.
However, there is also a countervailing force. Since an agent is not told about
whether $k=0$ or $k=1$ upon joining the queue, his belief about the initial queue length will also be updated as time progresses. On this account, he becomes  \emph{pessimistic} over time since the fact that he is still in the queue indicates that he likely underestimated the initial length of the queue when he
joined it. It turns out that the good news effect dominates the bad news effect in the $M/M/1$ model and even more broadly (see \cite{che2021optimal}). 

The mean residual wait time increases in other rules since the ``seniority'' in arrival order does not carry as much priority as in FCFS.  This is most evident with LCFS and LIEW.  As time passes, one faces a new entrant who takes away all or part of his priority, and one's residual wait time increases as time passes. SIRO suffers the same issue, albeit to a lesser degree.  Under SIRO, one's priority does not improve over time; only the queue length and one's belief about it matter. As time passes, one becomes more pessimistic on this account as he believes more buyers are in the queue than he initially thought, thus increasing his remaining wait time. 
  
At a more fundamental level, the above difference in queueing rules can be traced to the fact that they entail different {\it distributions} of wait times, although their  {\it mean} is the same.  In particular,   the wait time distribution is most ``fair,'' or {\it least dispersed}, under FCFS among all queueing rules; see \cite{Shanthikumar1987} for establishing this result under M/M/1 setup.  This means that both unusually short waits and unusually long waits are rare under FCFS compared with other rules. 
 
 This is easily seen in any realized within-cycle sample path of arrival times and departure times.   \Cref{fig:queue_comparison} illustrates an arbitrary sample path.\footnote{Any within-cycle sample path must have the same number of arrivals and departures.} A priority rule corresponds to a bipartite matching between arrival times and departure times.  In the $M/M/1$ model, the buyer who arrives first departs first under FCFS; this means that no two edges in the corresponding bipartite matching ``cross'' each other; see panel (a). By contrast, in any rule differing from FCFS, such as LCFS or  SIRO, the corresponding matching has crossing edges, as depicted in the panel (b).  This difference means a lower dispersion in wait times under FCFS; in the example, the wait times are $1.5, 1.5,$ and $2$ in (a), whereas they are $0.5, 0.5,$ and $4$ in (b), with the same mean $5/3$.\footnote{In general, since the wait time distribution becomes mean-preservingly contracted whenever any pair of crossed edges are swapped to reduce crossing, and since one can always repeat this swapping operation to move from any arbitrary matching to the matching under FCFS, it follows that the wait times under FCFS are a mean-preserving contraction of the wait times under any queueing rule.}

How does the wait time distribution affect a buyer's belief about residual wait time?  To study this, it is convenient to construct a probability space, as well as a buyer's belief, in two steps: (i) a sample path of arrival and departure times is first drawn according to exponential distributions (e.g., \Cref{fig:queue_comparison}) and (ii) a buyer's arrival time is then assigned to one of the arrival times (e.g., across $a_1$, $a_2$ or $a_3$).\footnote{The ``calendar'' time of one's arrival or departure carries no information in the steady state.}  One can then ``couple'' the sample paths under two queueing rules, such as (a) and (b).

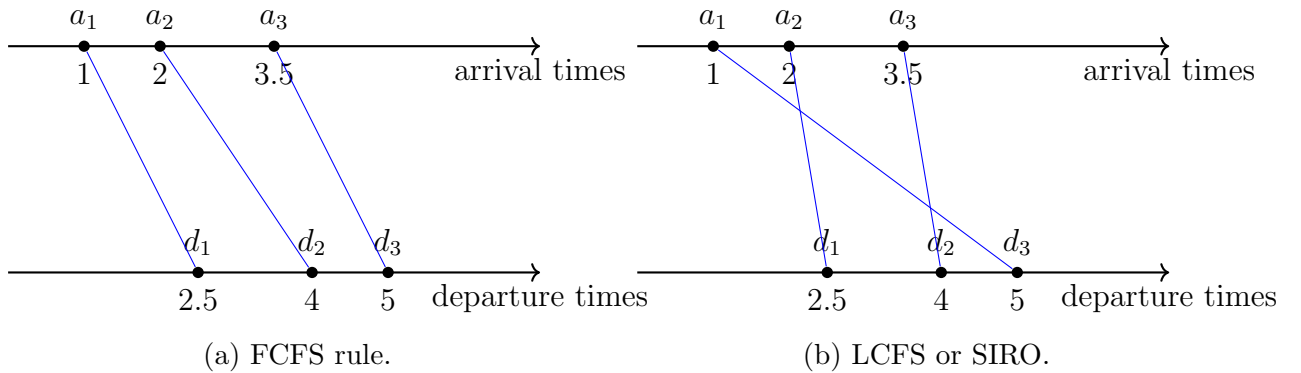
\begin{figure}[h!]
\centering

\begin{subfigure}{0.48\textwidth}
    \centering
    \begin{tikzpicture}
        \tikzstyle{agent}=[circle, fill, inner sep=1.5pt]

        \draw[->, thick] (0, 3) -- (7, 3) node[below] {arrival times};
        \draw[->, thick] (0, 0) -- (7, 0) node[below] {departure times};

        \node[agent, label=above:$a_1$, label=below:1] (a1) at (1, 3) {};
        \node[agent, label=above:$a_2$, label=below:2] (a2) at (2, 3) {};
        \node[agent, label=above:$a_3$, label=below:3.5] (a3) at (3.5, 3) {};

        \node[agent, label=above: $d_1$, label=below: 2.5] (d1) at (2.5, 0) {};
        \node[agent, label=above: $d_2$, label=below: 4] (d2) at (4, 0) {};
        \node[agent, label=above: $d_3$, label=below:5] (d3) at (5, 0) {};
        
        \draw[blue] (a1) -- (d1);
        \draw[blue] (a2) -- (d2);
        \draw[blue] (a3) -- (d3);
    \end{tikzpicture}
    \caption{FCFS rule.}
    \label{fig:fcfs_rule_updated}
\end{subfigure}
\hfill 
\begin{subfigure}{0.48\textwidth}
    \centering
    \begin{tikzpicture}
        \tikzstyle{agent}=[circle, fill, inner sep=1.5pt]

        \draw[->, thick] (0, 3) -- (7, 3) node[below] {arrival times};
        \draw[->, thick] (0, 0) -- (7, 0) node[below] {departure times};

        \node[agent, label=above:$a_1$, label=below:1] (a1) at (1, 3) {};
        \node[agent, label=above:$a_2$, label=below:2] (a2) at (2, 3) {};
        \node[agent, label=above:$a_3$, label=below:3.5] (a3) at (3.5, 3) {};

         \node[agent, label=above: $d_1$, label=below: 2.5] (d1) at (2.5, 0) {};
        \node[agent, label=above: $d_2$, label=below: 4] (d2) at (4, 0) {};
        \node[agent, label=above: $d_3$, label=below:5] (d3) at (5, 0) {};
        
        \draw[blue] (a1) -- (d3);
        \draw[blue] (a2) -- (d1);
        \draw[blue] (a3) -- (d2);
    \end{tikzpicture}
    \caption{LCFS or SIRO.}
    \label{fig:lcfs_siro_rule_updated}
\end{subfigure}

\caption{Comparison of queueing disciplines in a sample path.}
\label{fig:queue_comparison}
\end{figure}

Recall we are considering an optimal information policy, so a buyer does not observe the queue state when he joins the queue and thereafter.  Suppose a buyer forms his belief on (ii) for each arbitrary sample path (i) he considers possible, e.g., a path like the one portrayed in the example.\footnote{The explanation below conditions on each sample path. In the original problem, a buyer does not observe the realized sample path, so he also updates his belief on the sample path based on the elapsed time. In this sense, the explanation is intended to provide intuition on the force at work rather than a precise proof, which explicitly studies the evolution of beliefs based on one's queue position; see \Cref{app:che-tercieux} for details.}
Having spent time $t<0.5$ in the queue, the remaining expected wait time is $(5/3)-t$ for both (a) and (b), conditional on the illustrated sample path.  But after spending time $t\in [0.5, 1.5)$, the remaining wait times diverge under two rules: it is  $(5/3)-t$ in (a), whereas it is $4-t$ in (b), as the buyer infers his arrival time to be $a_1$.\footnote{For $t\in [1.5, 3]$, then the remaining wait time is $3-t$ in (a), and it is $4-t$.}   


Intuitively, dispersed waiting times make one pessimistic over one's residual wait times, worsening one's dynamic incentives.  As time passes, the fact that one {\it still remains in the queue} indicates that he has ``missed the early breaks'' and, therefore, the residual wait will be longer.  The fairness property of FCFS alleviates this problem, enabling the designer to implement the optimal outcome.  

\subsection{Large Market Limit}

Suppose the market becomes dense as $\l,\m\to\infty$ with the balance parameter $\rho=\l/\m$ held constant. The limit of such markets is captured by a  ``fluid'' model with a unit mass of items and a mass $\r=\l/\m$ of buyers arriving at every instant.\footnote{The term 'fluid model' (or 'fluid limit') is standard in queueing theory, operations research, and the study of stochastic networks to describe a deterministic approximation of a system where discrete arrivals and departures are smoothed into a continuous flow.} 

In such large markets, the stochasticity of the arrival processes and any uncertainty facing the designer disappear. Hence, the fundamental reason for holding queues of buyers or items no longer exists.  Yet, the incentive problem recognized by \cite{naor1969regulation} manifests itself extremely.  Recall that, given complete information, no admissions control, and FCFS, buyers will queue up to a level $K_{\fcfs}=\left\lfloor \frac{v}{c}\m \right\rfloor.$ 

Let's focus on the interesting case of  $\r\ge 1$.\footnote{If $\r<1$,  this maximal queue length is never binding, as $\l,\m\to \infty$, the excess supply absorbs demands with probability one, and the wait time shrinks to zero in probability.  }
 In this case, there is a perennial excess demand, and all buyers queue up to the $K_{\fcfs}$ level.  In the limit, the normalized average queue length, normalized by $\m$, and each buyer's wait time, converge to $\r\frac{v}{c}$ and $\frac{v}{c}$, respectively, in probability.  In short, all buyers queue up excessively to a degree that leaves them with no surplus.

 \Cref{thm:che-tercieux} remains valid even in this large market setting. Yet, the queueing rule/priority rule becomes irrelevant since the distribution of waiting time becomes degenerate. 
 Recalling $\r\ge 1$,  one can ensure that virtually all goods will be allocated to buyers with vanishing delays in the limit, by setting $K^*=(1-\epsilon)\m$, for an arbitrarily small $\epsilon>0$. The idea is that the small $\epsilon$ here creates states of excess supply and eliminates delay for those admitted into the queue. In particular, in the fluid model, we have the following.

 \begin{theorem} \label{thm:ntu0} Consider a continuum economy in which a unit mass of goods and a mass $\r\ge 1$ of buyers arrive at each instant. 
\begin{itemize}
    \item [(i)] Given the setup of \cite{naor1969regulation} with complete information, FCFS, and no admissions control,  all buyers wait for $v/c$ time to be served. The (per-unit time) total surplus is 0.

\item [(ii)] An optimal mechanism allocates (virtually) all goods to a queue of buyers whose length is capped so that no delay occurs.
\end{itemize}
\end{theorem}

Instead of a binding queue cap, the optimal policy (ii) can also be implemented by assigning a lottery that awards a unit of good with probability $ (1/\rho)-\epsilon$, under FCFS with complete information.

\subsection{Screening buyers with heterogeneous values}

The large market limit introduced above provides a convenient segue to study how waitlist/queueing can be used to allocate goods to buyers with heterogeneous values.  \cite{ashlagi2025optimal}, \cite{arnosti2020design}, and \cite{castro2021randomized} study a fluid (or large market model) model in which the designer allocates items with heterogeneous qualities to buyers with different values via some form of waitlist policies.\footnote{\cite{mekonnen2019random} develops a closely-related model with two-sided matching with frictional search. The trade-off between random search and directed search, the focus of this work, mirrors the trade-off between screening and random allocation discussed here. \cite{su-zenios2004} and \cite{schummer2021influencing} study a related issue in non-fluid models, but without the type of heterogeneity considered here.  Indeed, I am not aware of any NTU analysis that allows for value heterogeneity in the canonical queueing (i.e., non-fluid) setup, and this remains an open area of research. Even though \cite{leshnoAER} and \cite{baccara2020optimal} have two types of agents and items, the analysis is virtually isomorphic to a one-type model. \cite{che2025optimal} studies a fully general model of the value heterogeneity, but in a TU setup.}  

Here, I present a simplified version of \cite{ashlagi2025optimal} to present the main insights from these papers.  Consider a fluid model in which a unit mass of items and a mass $\r=\l/\m>1$ of buyers arrive at every instant. For the current purpose, it is convenient to interpret the model alternatively so that a mass $\r$ of items of heterogeneous qualities arrive at each instant, out of which mass 1 of items has high quality equal to one, and a mass $\r-1$  has low quality equal to zero. The real new feature is that the buyers have heterogeneous values $v$ of the item, distributed from $[0,1]$ according to a CDF $F$, which admits a strictly positive density $f$ on $(0,1)$. Assume that the inverse hazard rate $\frac{1-F(v)}{f(v)}$ is decreasing in $v$ for all $v\in (0,1)$.   

\cite{ashlagi2025optimal} considers discrete time (as opposed to the continuous time considered here) and allows for finite quality levels, rather than the two levels, 0 and 1, as in the current model. Nevertheless, the current model captures the central economic insights of the paper. The benefit is that we maintain the workhorse framework presented in this survey.\footnote{The multiple quality version required the authors to express feasible allocation---defined as a mapping from buyer values to the expected quality---as a mean-preserving contraction of the most positively assortative matching allocation.  They then use \cite{kleiner2021extreme}'s characterization of extreme points of such (majorized) allocations to argue that the optimal solution takes the form of a certain ironed version of the most positively assortative matching. The current two-level quality model simplifies ironing. } The designer's objective could be social welfare or ``allocative efficiency,'' which accounts for the gross surplus, ignoring the buyers' wait costs.  

Without loss, the designer specifies the eventual allocation probability $X(v)$ and the expected wait time $W(v)$, for the buyer who has just arrived and reported value $v$ at each instant. Since these are equilibrium objects, the pair must satisfy:
\begin{align*}
  U(v)&:=vX(v)-c W(v)\ge 0,\, \forall v; \tag{$IR_{\ntu}$} \label{IR'}  \\
  U(v)&\ge vX(v')-c W(v'),\, \forall v, v'; \tag{$IC_{\ntu}$} \label{IC'}  
\end{align*}
These specifications of \cref{IR'} and \cref{IC'} already portend the possible role of ``waiting costs'' as a screening device, often reserved for monetary transfers in the TU model.  This analogy is not accidental; the wait costs will play the same role as transfers, except for one important difference: the wait costs entail welfare costs whereas monetary transfers would not. 

As usual, one can
use the standard envelope characterization to replace \cref{IC'} and \cref{IR'} by:
\begin{align*}
 U(v)=\int_0^v X(s)ds  \, \forall v; \tag{$E_{\ntu}$} \label{E'}  \\
X(\cdot) \mbox{ is nondecreasing}. \tag{$M_{\ntu}$} \label{M'}  \\
\end{align*}


Next, let $p(v)$ denote the steady-state mass of buyers with values above $v$ in the queue. (This corresponds to the stationary distribution in our framework, even though $p(0)$, the total mass of buyers in the queue, need not equal 1.)  Then, we must have
\begin{align*}
 \r \int_v^1 W(s)f(s)ds &= p(v), \, \forall v; \tag{$L_{\ntu}$} \label{L'}  \\
  \r \int_0^1 X(s)f(s)ds &\le  1.\tag{$RF_{\ntu}$} \label{RF'}  \\
\end{align*}

The condition \cref{L'} follows from Little's law, which, for each $v$, relates the total wait time for incoming buyers with values above $v$ to the average queue length of these types, which in our continuum economy collapses degenerately to $p(v)$, the total mass of these types in the queue.  The condition \cref{RF'} ensures that the mass of promised allocation (after possible wait) equals the mass of items available at each instant.  (Recall that with our normalization, the normalized rate at which the items are received equals one.) This condition constitutes a Border style reduced-form auction characterization; see \cite{che2025optimal} as well as the next section for further details.

\cite{ashlagi2025optimal} consider allocative efficiency as an objective, in which case the problem is:
$$\max_{X, W, p}\, \r \int_0^1 vX(v)f(v)dv \eqno{[\mathcal P_{\ntu}'']}$$
$$\mbox{subject to } \quad \cref{IR'}, \cref{IC'}, \cref{L'}, \mbox{ and } \cref{RF'},$$

\begin{theorem} \label{thm:ntu} Assume $\r>1$. At the optimal mechanism solving $[\mathcal P_{\ntu}'']$,  all buyers with value  $v\ge \tilde v$ wait  for  $\tilde W:=\tilde v/c$ amount of time, where $\r[1-F(\tilde v)]=1$.  More formally, $X(v)=\mathbf{1}_{\{v\ge \tilde v\}}$, $W(v)=\mathbf{1}_{\{v\ge \tilde v\}}\tilde W$, and $p(v):=\r\frac{\tilde v}{c}[1-F(\max\{v, \tilde v\})].$   
\end{theorem}
\begin{proof}  Consider a further relaxed problem in which \cref{M'} and \cref{L'} are absent. Consider the Lagrangian (ignoring the constraint $X(\cdot)\in [0,1]$):  $$L=\int_0^1( v-\zeta )X(v)f(v)dv,$$
where $\zeta\ge 0$ is the multiplier for \cref{RF'}. Letting $\tilde v:=\inf\{v: v\ge \zeta\}$, we must have $X(v)=\mathbf{1}_{\{v\ge \tilde v\}}$, which satisfies \cref{M'}. If $\zeta=0$, then $X(v)=1$ for all $v$, which violates \cref{RF'} since $\r>1$, so \cref{RF'} is binding and $\tilde v$ is  pinned down by $\r[1-F(\tilde v)]=1$.   By setting $W(v)=\mathbf{1}_{\{v\ge \tilde v\}}\tilde v/c$, \cref{E'} and \cref{L'} are satisfied.  The chosen set of solutions so far is feasible and satisfies the complementary slackness condition. So, by the weak duality, it is an optimal solution to $[\mathcal P_{\ntu}'']$.
\end{proof}

Allocative efficiency requires allocating goods to buyers with values above a market-clearing ``price" $\tilde {v}$.  In the NTU context, the price can only be paid in waiting costs, requiring a wait time of $\tilde W:=\tilde v/c$, which is supported by the steady-state queue length of $\r [1-F(\tilde v)]\tilde v/c$. Facing this queue length and the requisite wait time, buyers enter if and only if $v\ge \tilde v$.  Note that the expected waiting time is pinned down by Little's law, regardless of the queueing discipline. One simple rule could be FCFS;\footnote{Other priority rules also work in the fluid model, since the residual mean wait time is degenerate at $\tilde v/c$, regardless of the queueing rule.} the allocation can be seen as heterogeneous-value version of \cite{naor1969regulation}, or \Cref{thm:ntu0}-(i).  

Panel (a) of \Cref{fig:welfare_comparison} depicts the allocatively efficient outcome when $\rho=1.5$ and $F$ is uniform.  Under the alternative interpretation with masses 1 and $\r-1$ of high-quality and zero-quality goods, the allocatively efficient outcome can be implemented by an FCFS with a deferral right: namely, buyers offered zero-quality goods can refuse assignment without losing their spots. 

It is important to recognize allocative efficiency doesn't correspond to social welfare maximization---an important difference relative to the TU setup that will follow. \cite{ashlagi2025optimal} also considers social welfare maximization.  Here, we consider an objective similar to the one considered in the earlier section:
$$\max_{X, W, p}\, \a \r\int_0^1 U(v)f(v)dv+(1-\a)\r \int_0^1 \pi X(v)f(v)dv \eqno{[\mathcal P_{\ntu}''']}$$
$$\mbox{subject to } \quad \cref{IR'}, \cref{IC'}, \cref{L'}, \mbox{ and } \cref{RF'},$$
where $\a\in [0,1]$ and $1-\a$ are respectively the weights the designer assigns to buyer welfare and firm profit ($\pi$ is generated whenever a buyer is assigned/served).

Substituting \cref{E'} into the objective function and simplifying the objective function allows us to rewrite the problem as:
$$\max_{X, W, p}\, \r\int_0^1 K(v) X(v)f(v)dv$$
$$\mbox{subject to } \quad \cref{RF'},  \mbox{ and } \cref{M'},$$
where $K(v):=\a \frac{1-F(v)}{f(v)}+(1-\a)\pi.$ 
Given our assumption of decreasing inverse hazard rate, $K(\cdot)$ is nonincreasing for all $\a$.

\begin{theorem} \label{thm:ntu2} Assume $\r>1$.  At the optimal mechanism solving $[\mathcal P_{\ntu}''']$, all buyers are served with a lottery $X(\cdot)\equiv 1/\r$ without any delays (i.e., $W(v)=p(v)\equiv 0).$
\end{theorem}
\begin{proof}  Fix any nondecreasing $X(\cdot)$ satisfying \cref{RF'}. Let $\bar x:=\int_0^1 X(v)f(v)dv\le 1/\r$.  Then, since $K(\cdot)$ is nonincreasing and nonnegative,  by the Cauchy-Schwarz inequality, 
$$\int_0^1K(v)X(v)f(v)dv\le \left(\int_0^1 K(v)f(v)dv\right)\bar x \le \left(\int_0^1 K(v)f(v)dv\right) \frac{1}{\r},$$
where the last inequality follows from \cref{RF'}.
\end{proof}

Allocative efficiency corresponds to social welfare maximization in the TU setting; however, this is not the case in the NTU setup, since the price paid by buyers to claim objects is not a transfer but a wasteful social cost.  The dark blue area in Panel (a), therefore, constitutes the only net social surplus; the red area is dissipated through waiting costs.

When the buyers are homogeneous, such waste can be eliminated through admissions control or a lottery, as shown in \Cref{thm:ntu0}-(ii). However, doing so is costly here since the designer must sacrifice allocative efficiency.  In principle, it is unclear how the designer would trade off these two objectives.  Yet, given the declining inverse-hazard rate condition, the social optimum ($\a=1$) completely sacrifices allocative efficiency to minimize waiting costs.  In other words,  complete pooling with random allocation is prescribed, which entails no delay.  In the example, the welfare under full screening is $2/9$ (the dark blue area in (a)), whereas the welfare under pooling is $1/2$ (the light blue area in (b) times $1/1.5$).  
Under the alternative interpretation with a mass $\r-1$ of zero quality good and mass 1 of high quality good, this complete pooling can be implemented with an FCFS without a deferral right; namely, buyers offered zero quality goods can't refuse the assignment.  See \cite{arnosti2020design} and \cite{castro2021randomized} for proposing queuing mechanisms with a similar feature. 

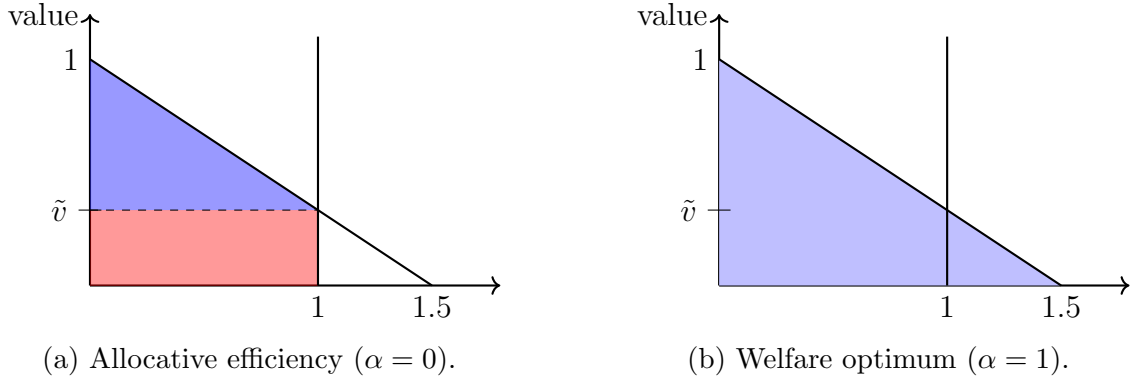
\begin{figure}[h!]
\centering
\def\rho{1.5}

\begin{subfigure}{0.48\textwidth}
    \centering
    \begin{tikzpicture}[scale=3]
        \draw[->, thick] (0,0) -- (1.8,0) ;
        \draw[->, thick] (0,0) -- (0,1.2) node[left] {value};

        \coordinate (A) at (0,1);
        \coordinate (B) at (\rho, 0);
        \coordinate (C) at (1, 0);
        
        \pgfmathsetmacro{\yintersect}{1 - 1/\rho}
        \coordinate (D) at (1, \yintersect);
        \coordinate (E) at (0, \yintersect);

        \fill[red, opacity=0.4] (0,0) rectangle (D);
        \fill[blue!80, opacity=0.5] (A) -- (D) -- (E) -- cycle; 

        \draw[thick] (A) -- (B); 
        \draw[thick] (1,0) -- (1,1.1); 
        \draw[dashed] (E) -- (D); 
        
        \node[below] at (C) {1};
        \node[left] at (A) {1};
        \node[below] at (B) {$\rho$}; 
        \draw (0.05, \yintersect) -- (-0.05, \yintersect) node[left] {$\tilde v$};
    \end{tikzpicture}
    \caption{Allocative efficiency ($\a=0$).}
    \label{fig:allocative_efficiency_final}
\end{subfigure}
\hfill 
\begin{subfigure}{0.48\textwidth}
    \centering
    \begin{tikzpicture}[scale=3]
        \draw[->, thick] (0,0) -- (1.8,0) ;
        \draw[->, thick] (0,0) -- (0,1.2) node[left] {value};

        \coordinate (A) at (0,1);
        \coordinate (B) at (\rho, 0);
        \coordinate (C) at (1, 0);

        \fill[blue!25] (0,0) -- (A) -- (B) -- cycle; 
        
        \draw[thick] (A) -- (B); 
        \draw[thick] (1,0) -- (1,1.1); 

        \pgfmathsetmacro{\yintersect}{1 - 1/\rho}

        \node[below] at (C) {1};
        \node[left] at (A) {1};
        \node[below] at (B) {$\rho$}; 
        \draw (0.05, \yintersect) -- (-0.05, \yintersect) node[left] {$\tilde v$};
    \end{tikzpicture}
    \caption{Welfare optimum ($\a=1$).}
    \label{fig:welfare_optimum_final}
\end{subfigure}

\caption{Comparison of welfare measures.}
\label{fig:welfare_comparison}
\end{figure}

The main insight of \Cref{thm:ntu2} is traced back in its intellectual provenance to \cite{mcafee1992bidding}, who showed that a bidding ring seeking to maximize its members' welfare would prefer to assign a winner at random when a knock-out action is infeasible (so transfers can't be used), given the same inverse hazard rate condition,\footnote{See also \cite{che2018weak} for the optimality of pooling in the auction environment, for the same reason.} and to \cite{condorelli2012money} and \cite{hartline2008optimal}, who used the same logic to argue that complete pooling is socially optimal when the designer must rely on costly signaling to screen agents.  The difference here is that costly signaling must take the form of waiting in a queue, which, in the steady state, must require the queue length to be adjusted endogenously to satisfy incentive compatibility. The current overview also reveals the welfare cost of the queue as a manifestation of \cite{naor1969regulation}'s problem and pooling as its remedy.  

\section{Transferable Utility Model}

We now turn to the Transferable Utility (TU) model. While the previous section focused on markets where prices play a limited role in allocation, many dynamic environments utilize monetary transfers to balance competition and manage supply and demand. Following the framework established in \cite{che2025optimal}, we generalize the classic mechanism design paradigm of \cite{myerson1981} to the dynamic setting with asynchronous and stochastic arrivals of both items and buyers.

\subsection{Setup}
The model is the same as before in its basic primitives:  in a continuous time $t\ge 0$, the designer(platform)  receives units of a homogeneous good arriving at a Poisson rate of $\mu>0$. Buyers with unit demand arrive at a Poisson rate of $\lambda>0$.  What is different, however, and similar to \cite{myerson1981}, is that buyers have heterogeneous values for the good.  Specifically, each buyer has value $v\in [0,1]$, drawn independently for each buyer at the time of his arrival according to a CDF $F$ with a density $f$ that is strictly positive and absolutely continuous. We assume that the virtual value $J(v)\:=v-(1-2\a)\frac{1-F(v)}{f(v)}$ is nondecreasing in $v$ for all $\a$.\footnote{This condition holds when $F$ is uniform for all $\alpha \in [0, 1]$. More generally, for distributions with a non-increasing inverse hazard rate, the condition is always satisfied for $\alpha \in [0, 1/2]$. For $\alpha > 1/2$, the condition requires that the inverse hazard rate does not decrease too rapidly; specifically, we require $1 - (1-2\alpha)\frac{d}{dv}[\frac{1-F(v)}{f(v)}] \ge 0$. If this strong form of regularity is violated, then I conjecture that a form of ironing suggested by \cite{myerson1981} will apply.}  Each buyer is privately informed of his value $v$.

As before, the designer may store buyers in a queue, at $c>0$ per unit time per buyer, if there are no items available, and she can store goods in inventory at cost $d>0$ per unit time per item, in case there are no buyers. 

Unlike the NTU model, buyers have transferable utilities that are linear in money.   So, if a buyer with value $v$ spends time $t$ in the queue, makes a monetary payment of $y$, and receives the good, his payoff is $v-y-ct$. If he spends time $t$ in the queue and pays $y$ but does not receive the good, his payoff is $-y-ct$.

\paragraph{Mechanisms.}   We again focus on regenerative mechanisms, as defined in the NTU model. The TU setting introduces two key differences. First, buyers have heterogeneous private values $v \in [0, 1]$. By the (Bayes-Nash) revelation principle, we can restrict our attention to direct mechanisms $\phi$ that condition on buyers' reported values. Second, the mechanism's set of outcomes is expanded to include monetary payments. Therefore, a mechanism $\phi$ is a non-anticipatory, measurable mapping from histories (which include primitive arrivals and all reported values) to outcomes (queuing decisions, matches, and payments).

This richer setting necessitates a more complex state space. A queue state $\theta$ must now capture not only the queue length $k \in \mathbb{Z}$ (where $k<0$ denotes item inventory) but also the reported values of buyers currently waiting. We represent this by an ordered vector $\mathbf{v}=(v_{1},v_{2},....) \in [0,1]^{\mathbb{N}}$ of reported values. Let $\mathcal{V} \subset [0,1]^{\mathbb{N}}$ be the set of all such ordered vectors. We continue to impose the No Allocation Delay (NAD) condition. The full queue state space is thus $\Theta := (\cup_{l \in \mathbb{N}} \{-l\}) \cup \{0\} \cup \mathcal{V}$.

As before, we restrict attention to \textbf{Positive-Recurrent Regenerative Mechanisms (PRRMs)}. Any such mechanism $\phi$ induces a positive recurrent process on the state space $\Theta$, which admits a unique stationary distribution $p \in \Delta(\Theta)$. We require the mechanism to be incentive-compatible. Let $\Phi$ denote the set of all PRRMs in which buyers have incentives to follow all recommendations, including reporting their values truthfully, as a Bayes-Nash equilibrium. This equilibrium is evaluated assuming an incoming buyer's prior on the queue state is given by the stationary distribution $p$ induced by the mechanism.

The problem facing the designer is:
$$\sup_{\phi\in \Phi}  \int_{\t} \int_0^1\Big[\alpha \l U^M(v; \t)+(1-\a)\l T^M(v;\t)-\sum_{\ell\in \N} \mathbf{1}_{\{\t=-\ell\}} \ell d\Big]f(v)\text{d}v \bp(\text{d}\t),\eqno{[\mathcal P_{\tu}]}$$
where $\a$ is the weight on the buyer welfare,
$U^M(v; \t)$ and $T^M(v; \t)$ are respectively the expected utility and payment for a buyer with value $v$ arriving in state $\t$ under mechanism $\phi$.  The problem specializes to revenue maximization when $\a=0$ and to the welfare maximization, or equivalently allocative efficiency, when $\a=1/2$.

\paragraph{Relation with Literature.}
 This model generalizes the static mechanism design framework, as presented in \cite{myerson1981}, by considering a dynamic environment, which is more descriptive of platform-mediated modern marketplaces. In such a setting, the key role of the designer is not just to allocate items to buyers when they are all present simultaneously, but also to manage and transfer competition across time by judiciously storing buyers or goods. 
  
 Importantly, however, this model does not nest the NTU model surveyed earlier. Transferability makes it easy for the designer to control buyers'  incentives. For example, dynamic incentives can be managed simply by reimbursing buyers for the waiting costs they incur based on the amount of time they spend in the queue. Recall that most of the analytical challenge in the previous section resulted from the difficulty associated with managing buyers' incentives to queue in the NTU model; this issue becomes easier to handle in the TU setup.

\subsection{Optimal Mechanism}  The problem $[\mathcal P_{TU}]$  is difficult to solve. Instead, a relaxed problem is set up as follows. 

Fix any mechanism $M\in \M^{**}$.  It induces an interim allocation probability $X(v)$ and payment $T(v)$ for a buyer who has just arrived and reported value $v$, where the payment is made after (or net of) the reimbursement of waiting costs, which we assume in the sequel.\footnote{Recall we already observed that the designer may, without loss, reimburse waiting costs.} Since $M$ is incentive compatible, they must satisfy:.\footnote{Note that \cref{IC} ignores possible double deviations, so it is necessary but not sufficient for buyers to report truthfully.  This is not a problem since we are considering a relaxed program.}
\begin{align}
   U(v)&  := vX(v) - T(v) \geq 0,\, \forall v, \tag{$IR$} \label{IC}\\
    U(v) &\geq vX(v') - T(v'), \, \forall v,v'. \tag{$IC$} \label{IR}
\end{align}

Obviously, not all $X(v)\in [0,1]$ is feasible. The allocation promised to a buyer must be compatible with the stochastic supply of goods as well as with the promises the designer makes to buyers who arrived before and those who will arrive in the future. To handle the feasibility issue, we first let  $p_k(v)$ denote the probability that  {\it exactly $k$ buyers have values strictly above $v$ in the steady  state}. Accordingly, 
$p_k(0)$ denotes the probability that there are exactly $k$ buyers in the steady state queue.  Also, as before, let $p_{-\ell}$, $\ell\in \N$ denote the stationary probability that there are $\ell$ items in the inventory.  

Then, feasibility of $X(\cdot)$ requires:
\begin{align*}
   \lambda  \int^1_v X(s)f(s)ds
   \le \mu\sum_{k=1}^{\infty} y_k(v)p_{k}(v)+ \sum_{\ell=1}^{\infty}p_{-\ell} \int_v^1 z_{\ell}(v)f(v)dv, \quad \forall v\in V, \tag{$RF$} \label{RF}
\end{align*}
where  $y_k(v)$ is the probability that an incoming good is allocated to one of the $k$ buyers with values above $v$ and $z_{\ell}(v) $ is the probability with which an incoming buyer with value $v$ is allocated the good when there are $\ell$ items in storage.  The LHS describes total allocation promises made to types above $v$ per unit of time, while the RHS describes total allocation made to the buyers with values above $v$ per unit time, noting that an allocation occurs when a good arrives with buyers waiting in the queue or when a buyer arrives with goods in storage.  This condition is similar in spirit to \cite{border1991implementation} and \cite{che2013generalized} but has an added temporal dimension.
 
We next turn to the feasibility of the distribution $\bp=(p_k)$.  Since these are stationary objects, they must respect balance conditions.  For a subset $\Z_{\le -\ell} \subset \T$, for $\t=-\ell$, for some $\ell\in\N$, we must have a balance condition between transitions between $\Z_{\le -\ell}$ and $\T\setminus \Z_{\le -\ell}$:
\begin{align*}
\m p_{-\ell} \ge \l p_{-(\ell+1)}, \tag{${B_\ell}$} \label{Bl}   
\end{align*}
somewhat analogously to the condition we had in the NTU setup.  What is new and potentially difficult is the balance condition for each measurable set $\V' \subset \V \subset \T$, namely the set of profiles of values reported by the buyers in the queue. To this end, we only focus on (measurable) subsets of $\V$ of the form:
$$\V_k(v):=\{\bv : v_{k+1}\le v\},$$
which comprises a set of queue states in which the $k+1$-st highest value is less than $v$, or equivalently, at most $k$ buyers have values above $v$.   \Cref{fig:events-V_k(v)} depicts $\V_0(v)$ and $\V_1(v)$ (only in the first two coordinates).

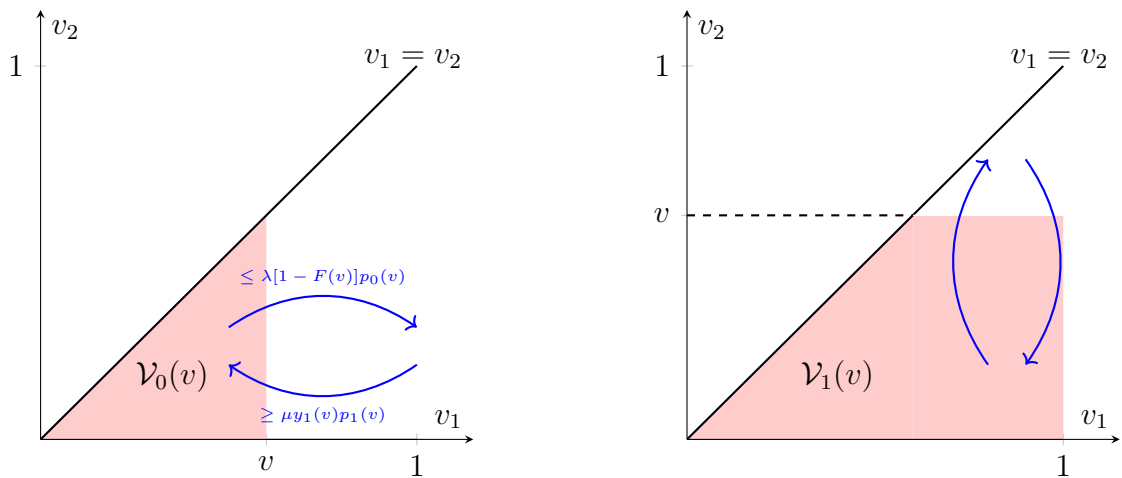
\begin{figure}[h!]
\centering
\begin{tikzpicture}
\begin{groupplot}[
    group style={
        group size=2 by 1, 
        horizontal sep=80pt, 
    },
    xlabel={$v_1$},
    ylabel={$v_2$},
    xlabel style={below}, 
    ylabel style={rotate=90, above}, 
    xmin=0, xmax=1.15, 
    ymin=0, ymax=1.15, 
    axis lines=middle,
    axis equal,
    unit vector ratio*=1 1 1,
    xtick={0,1}, 
    ytick={0,1}, 
    grid=none, 
]

\nextgroupplot
\addplot[name path=A, domain=0:1, samples=100, color=black, thick] {x}; 
\addplot[name path=B, domain=0:1] {0}; 
\addplot [red!20] fill between[of=A and B, soft clip={domain=0:0.6}]; 
\node at (axis cs:0.35,0.17) {$\V_0(v)$}; 
\node at (axis cs:0.99,1.02) {$v_1 = v_2$}; 
\pgfplotsset{xtick={0,1,0.6}, xticklabels={0,1,$v$}} 

\draw[->, blue, thick, bend left=35] (axis cs:0.5, 0.3) to node[above, sloped, text=blue] {\tiny $\leq \l[1-F(v)]   p_{0}(v)$} (axis cs:1, 0.3);
\draw[->, blue, thick, bend left=35] (axis cs:1, 0.2) to node[below, sloped, text=blue] {\tiny $\geq \mu y_{1}(v) p_1(v)$} (axis cs:0.5, 0.2);

\nextgroupplot
\addplot[name path=C, domain=0:1, samples=100, color=black, thick] {x}; 
\addplot[name path=D, color = white, domain=0.6:1] {0.6}; 
\addplot[name path=E, domain=0:1] {0}; 
\addplot [red!20] fill between[of=C and E, soft clip={domain=0:0.6}]; 
\addplot [red!20] fill between[of=D and E, soft clip={domain=0.6:1}]; 
\node at (axis cs:0.40,0.17) {$\V_1(v)$}; 
\node at (axis cs:0.99,1.02) {$v_1 = v_2$}; 
\pgfplotsset{ytick={0,1,0.6}, yticklabels={0,1,$v$}} 
\addplot [color=black, dashed, thick] coordinates {(0, 0.6) (0.6, 0.6)}; 

\draw[->, blue, thick, bend left=35] (axis cs:0.8, 0.2) to node[pos=0, above, text=blue, xshift=-42pt, yshift=-5pt] {} (axis cs:.8, 0.75);
\draw[->, blue, thick, bend left=35] (axis cs:.9, 0.75) to node[below, text=blue] {} (axis cs:0.9, 0.2);

\end{groupplot}
\end{tikzpicture}
\caption{$\V_0(v)$ and $\V_1(v)$}
\label{fig:events-V_k(v)}
\end{figure}

A balance condition for $\V_k(v)$ boils down to: 
\begin{align*}
	\lambda[1-F(v)]   p_{k}(v)    \ge   \mu y_{k+1}(v) p_{k+1}(v).  \tag{${B_k}$}  \label{Bk}   
\end{align*}
The LHS is an upper bound on the outflow from set $
\V_k(v)$, depicted for $k=0$ by the out-arrow.  Ignoring the higher order terms, the outflow occurs when there are exactly $k$ buyers with values strictly above $v$---an event that occurs with probability $p_{k}(v)$---and a buyer with value above $v$ arrives---which occurs at rate $\l[1-F(v)]$.  In that case, a transition out of $\V_k(v)$ occurs if that buyer is admitted into the queue. Since he may not, the LHS gives an upper bound for the outflow.  Analogously, the RHS gives inflow into the set $\V_k(v)$.

These conditions here apply only to lower-dimensional subsets of measurable sets that characterize the stationary distribution $\bp\in \Delta(\T)$. To be more precise, $\V_k(v)$ is indexed by each $(k,v)$, so we have cardinality $\N\times [0,1]$ of conditions.  But to account for all measurable sets, we need a condition for a set of the form  $[\mathbf{0},\bv]$, for each $\bv \in \V$, so the cardinality of the conditions becomes in the order of $\N^{[0,1]}$, far bigger than $\N\times [0,1]$. Clearly, the selected conditions are necessary.  Although they are not sufficient for characterizing the stationary distribution, they turn out to be sufficient for identifying the optimal solution to $[\mathcal P_{\tu}']$ (that will be stated below); as only they will be seen to bind at the optimal solution. The reduction is ultimately what makes the problem tractable.  

We are now ready to formulate our relaxed program.   
Consider the problem:
\begin{align*}
    \max_{p,y,z, X,T} \l \int_0^1 [\a U^M(v)&+(1-\a) T(v)] f(v)dv- 
  c \sum_{k=1}^{\infty} k p_k(0)-d\sum_{\ell=1}^{\infty}\ell p_{-\ell} \\
\mbox{   subject to }& \,
\cref{IC}, \, \cref{IR},\, \cref{RF}, \cref{Bl}, \mbox{ and } \cref{Bk}, \,\, \forall \ell, k. 
\end{align*}
Again, one can interpret the objective as the long-run time average of the weighted sum of consumer and designer/producer surplus (with $\alpha$ being the weight for the former) minus the buyer waiting costs the designer reimburses and inventory costs, or its counterpart for steady-state flow surplus. Note that the term $\sum_{k=1}^{\infty} k p_k(0)$ accounts for the steady-state average queue length: $p_k(0)$ is the stationary probability that there are $k$ buyers in the queue (given our convention to encode absenece of a buyer by a presence of buyer with value $0$).

Clearly, all constraints are necessary for a mechanism to be in $\M^{**}$.  Using the standard envelope condition, the problem is further relaxed to:
$$\max_{p,y,z, X} \l \int_0^1 J(v)X(v)f(v)dv- 
	c \sum_{k=1}^{\infty} k p_k(0)-d\sum_{\ell=1}^{\infty}\ell p_{-\ell} \eqno{[\mathcal P_{\tu}']}$$
$$\mbox{ subject to } \, \cref{RF}, \cref{Bl}, \mbox{ and } \cref{Bk}, \,\, \forall \ell, k, $$
where $J(v):= v - (1-2\a) \frac{1-F(v)}{f(v)}$, which we assume is increasing in $v$.

Upon suitable change of variables, $[\mathcal P_{\tu}']$ can be seen as a linear program.  Note that we have transformed a dynamic mechanism design problem into a linear optimization problem.  What made this transformation possible is the combination of traditional mechanism design tools, such as \cref{IC} and \cref{IR}, with the reduced-form characterization \cref{RF} and the balance conditions necessitated by stationarity.  While the approach here is similar in spirit to its NTU counterpart, $[\mathcal P_{\ntu}']$, the presence of private information, and more importantly, the richness of the queue state space $\T$ sets it apart. As mentioned earlier, reducing the balance conditions for stationarity is a key step toward a tractable LP formulation.  The optimal solution to $[\mathcal P_{\tu}']$ is characterized next.


\begin{theorem}\label{thm:main}  The optimal solution to the relaxed program $[\mathcal{P}_{TU}^{\prime}]$ is characterized as follows:  There are $\hat v_K> \cdots >\hat v_1>\hat v_{-1}> \cdots > \hat v_{-L}>J^{-1}(0)$, for some $K, L\in \N$, such that (i) items are stored only up to $L$ units; (ii) buyers are queued up to  $k\le K$ buyers  if and only if all $k$ of them  have  values above $\hat  v_k$; (iii) if  a buyer arrives with $\ell\ge 1$ items in storage, he is allocated the good if and only if his value is above $\hat v_{-\ell} $; and (iv) if an item arrives with $k\ge 1$ buyers waiting in the queue, then it is assigned to the buyer with the highest value.
\end{theorem}
\begin{proof}  The proof involves identifying the set of primal and dual variables that satisfy the set of all complementary slackness.  The interested reader is referred to \cite{che2025optimal} for details.
\end{proof}

Similar to the static mechanism design, the optimal dynamic mechanism allocates the items optimally among (endogenously selected) participants.  Unlike the static mechanism, the key aspect of the optimal dynamic mechanisms concerns the design of queues. This aspect is crucial not only for balancing the intertemporal mismatch between demand and supply, but also for allocating competitive pressures across time to effectively discipline privately informed buyers.  Since both buyers and items are costly to store in queues, the mechanism requires buyers to meet progressively higher cutoffs as the queue length increases and the number of inventoried goods falls.

Note that the qualitative features of the optimal allocation are similar for revenue maximization ($\a=0$) and welfare maximization ($\a=1/2$).  The only difference between the two is that queue-dependent cutoffs $\hat v_k$ differ between the two cases.  In the revenue maximization case, they are chosen to maximize the virtual value  $J(v; \alpha=0)=v-\frac{1-F(v)}{f(v)}$, just as in the Myerson setting, whereas in the welfare maximization case, they are chosen to maximize realized value $J(v; \alpha=1/2)=v$.  Since the welfare-maximizing designer lacks the monopoly exclusion motive, one would expect the latter to be lower in general.  This is indeed true for a low queue size. For example, given $d=\infty$, the welfare maximizing cutoff for one buyer queue is $\hat v_1=c$, whereas the corresponding revenue-maximizing cutoff is $\hat v_1=J^{-1}(c)$, the same as the standard Myerson reserve-price with cost $c$.  More surprisingly, however, for a large queue size, the order is reversed:  the revenue-maximizing cutoffs are lower than the welfare-maximizing cutoffs. This single crossing feature can be explained as follows.  Again, lacking the monopoly exclusion motive, the welfare-maximizing designer is more willing to admit buyers into the queue when the queue is relatively short, thereby reducing the risk of buyer stockout (i.e., the situation where no buyers are available when a good arrives). 
At the same time, the revenue-maximizing designer is more willing to admit buyers into a long queue, since she is more keen on selling to high-value buyers who command lower information rent, and because the higher exclusion for shorter queues already means a stronger need to insure against the buyer stockout risk.
Consequently, compared to welfare-optimal thresholds, revenue-maximizing thresholds are higher for short queues and lower for long queues.

\Cref{thm:main} characterizes the optimal solution to $[\mathcal P_{\tu}']$. What remains is to show that the optimal solution to this relaxed problem can be made incentive compatible by a mechanism $M$ satisfying the constraints of our original program $[\mathcal P_{\tu}]$. \cite{che2025optimal} demonstrates that a mechanism comprising a series of auctions can implement the optimal mechanism in {\it dominant strategies.}

To illustrate how the optimal mechanism works, suppose a buyer arrives following a null state initially. Then, a designer imposes a reserve price of $\hat v_1$, a minimum price he will be charged for a good he may receive, making it dominant for the buyer to join the queue if and only if his value is above $\hat v_1$. Suppose another buyer arrives before an item arrives.  Then, the reserve price, or an ascending auction clock, rises from $\hat v_1$ continuously until one of them drops out when the price reaches his value, or until it reaches $\hat v_2$, whichever happens first.  In the first case, the price is stopped at the drop-out price, with the surviving buyer remaining in the queue.  In the latter case, both buyers survive and stay in the queue.  More generally, if a buyer arrives at a queue with length $k$, then a clock price will similarly rise from the existing stopped price until a buyer drops out or a price of $\hat v_{k+1}$ is reached, whichever occurs first.

Suppose a good arrives next with buyers waiting. Then, a sealed-bid auction is held in which buyers are required to make a bid no less than the highest reserve prices they have survived.  The highest bidder wins the good (with ties broken at random)  and is required to pay the {\bf cutoff price}---defined as {\it the lowest bid he could have made and would have eventually won in light of the future sample path, assuming that the same bid is used to determine future assignments.}   A cutoff price depends on the sample path realized (possibly well) after the auction, reflecting the realized competition the designer faces afterwards. In this dynamic setting, the price is not a single number known at the moment of matching, but a function of the realized stochastic process.\footnote{To fix the idea, consider a winner, $A$, who arrives when another buyer, $B$, is already in the queue. In one scenario, a second item arrives later to satisfy $B$; here, $A$'s payment is low because the eventual supply was sufficient for both. In a second scenario, no second item arrives, but two new high-value buyers arrive instead. In this case, $A$ must pay a higher price because, in light of that realized sample path, $A$ faced much tougher competition to secure one of the scarce items.}  In this case, the winner is ``billed'' after receiving the good.  

Intuitively, the cutoff prices track the future level of competition to discipline the current buyers. If the supply condition improves (with the arrival of more items)  following the assignment, a low price will be charged to the winner. By contrast, if excess demand arises, the winner will be charged a high price.   

Next, the designer caps the size of the inventory at some finite $L$ units; an additional item received after reaching $L$ is discarded.  Suppose there is an inventory of $\ell\ge 0$ items, and a buyer arrives.  Then, the buyer is charged a fixed price of $\hat v_{-\ell}$; as noted before, the price is lower the larger the inventory. 

The auction/pricing mechanism described here implements the optimal outcome in dominant strategies.  While we allowed for all PRRMs, the optimal DSIC mechanism is {\it pseudo-Markovian:} its allocation and queue/inventory decisions are Markovian, depending only on the queue state $\t\in \T$, whereas the cutoff price depends possibly on the entire within-cycle history.  An implication is that the allocation depends only on the reported values of the buyers in the queue, independently of their arrival orders; in other words, the allocation priority does not follow the standard queue discipline.

In summary, the optimal dynamic mechanism retains the core element of \cite{myerson1981} — namely, allocating goods to buyers who are already present via an auction with a reserve price. Furthermore, it involves an additional feature whereby the designer stores buyers or goods in a queue to balance the potential intertemporal mismatches between demand and supply, and, no less importantly, to balance competitive pressures across time.

\subsection{Large Market Properties} \label{sec:large-market}

Suppose the market becomes dense as $\l,\m\to\infty$ with the balance parameter $\rho=\l/\m$ held constant, or as $c$ or $d$ vanishes. As with the NTU model, the limit of such markets corresponds to a static model with a unit mass of items on the supply side and a mass $\r=\l/\m$ of buyers on the demand side (see \cite{che2025optimal}), or a dynamic model in which a unit mass of items and a mass $\r=\l/\m$ of buyers arrive at each instant.

The optimal mechanism in this limit continuum model is very simple.  Focusing on revenue maximization (i.e., $\a=0$), the optimal mechanism is simply the uniform-price multiunit auction with optimal reserve price: i.e., selling the mass of items at price equal to either the marginal buyer's value $\tilde v:=\inf\{v: \rho [1-F(v)]\le 1 \},$ or the standard monopoly price $\hat v_0:=J^{-1}(0)$, whichever is higher.  (The former is the continuum market analog of the highest losing bid in the uniform-price auction.)  Note also that, as $\a\to 1/2$, $\hat v_0\to 0$.

\Cref{fig:large-market} features the situation in which the former price is higher.

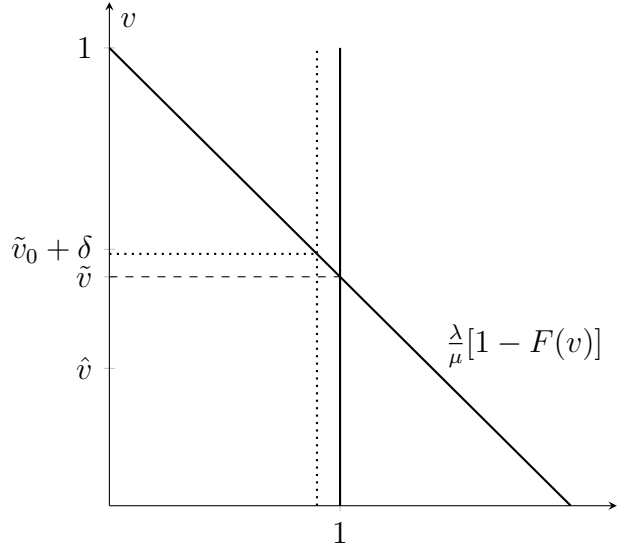
\begin{figure}[H]
\centering
\begin{tikzpicture}
\begin{groupplot}[
    group style={
        group size=1 by 1, 
        horizontal sep=40pt, 
    },
    height=20em,
    xlabel={},
    ylabel={$v$},
    xlabel style={below}, 
    ylabel style={rotate=90, above}, 
    xmin=0, xmax=1.1, 
    ymin=0, ymax=1.1, 
    axis lines=middle,
    axis equal,
    unit vector ratio*=1 1 1,
    xtick={0,1}, 
    ytick={0,1}, 
    grid=none, 
]

\nextgroupplot
\addplot[name path=A, domain=0:1, samples=100, color=black, thick] {1-x};
\addplot [name path = B, color=black, dashed] coordinates {(0, 0.5) (0.5, 0.5)}; 
\addplot [name path = B, color=black, dotted, thick] coordinates {(0, 0.55) (0.45, 0.55)};
\addplot[name path=C, domain=0:1, samples=100, color=black, thick] coordinates {(0.5,0) (0.5,1)};
\addplot[domain=0:1, samples=100, color=black, thick, dotted] coordinates {(0.45,0) (0.45,1)};
\node at (axis cs:0.9,0.35) {$\frac{\l}{\m}[1-F(v)]$};
\pgfplotsset{ytick={0,1,0.5, 0.56, 0.3}, yticklabels={0,1,$\tilde v$, $\tilde v_0 +\delta$, $\hat v$}} ;
\pgfplotsset{xtick={0,0.5}, xticklabels={0,1}};
\end{groupplot}
\end{tikzpicture}
\caption{Large market limit of the optimal mechanism}
\label{fig:large-market}
\end{figure}

\cite{che2025optimal} establishes the following results.

\begin{theorem}\label{thm:large-market} The normalized objective converges to the continuum model optimum if (a)
 $\l,\m\to \infty$ with $\rho=\l/\m$ held constant, or (b) $c \to 0$, or (c) $d \to 0$. 
\end{theorem}

\noindent {\it Proof Sketch:}  Multiplying the four parameters $\l,\m,c$, and $d$ by a constant $k > 0$ is equivalent to simply rescaling time.
Hence, (b) is the same as $\l,\m,d \to \infty$ with $\l/\m$ and $\l/d$ held constant, and (c) is the same as $\l,\m,c \to \infty$ with $\l/\m$ and $\l/c$ held constant. Thus (a) is implied by either (b) or (c).

To prove (b),  consider a simple feasible mechanism that assigns an incoming good to a waiting buyer if the buyer queue is non-empty and discards the good otherwise, and queues a buyer if and only if his value is above $\tilde v_0 + \delta$ for an arbitrarily small $\delta>0$ (see \cref{fig:large-market}). This ensures that the good arrival rate $\mu$ is greater than an ''effective'' buyer arrival rate $\l[1-F(\tilde v_0 + \delta)]$.  Then, every buyer with a value above $\tilde v_0 + \delta$ is eventually assigned a good after a delay; however, one can show that the delay vanishes as $ c \to 0$.  Therefore, the per-unit revenue converges to the level attained in the optimal multiunit auction mechanism in the continuum model. This indicates that the optimal mechanism must also converge to the latter, which can be shown to provide the upper bound in the limit of a large market. 

Similarly, to prove (c), we consider a mechanism that sells a stored item to an incoming buyer at a price $\tilde v_0+\delta$, turns away buyers when inventory is empty, and stores at most $\bar L:=\sqrt{\mu/d}$ items. As $d\to 0$, one can show both the average storage cost borne by the designer and the probability of stockout (conditional on buyer arrival) vanish.  \qed

One can see a stark yet obvious difference relative to the NTU case.  Recall from \Cref{thm:ntu2} that a social optimum in the NTU setting involves a complete pooling or random assignment; here, the social optimum ($\a=1/2$) yields allocative efficiency.

\section{Broader Overview and Future Directions}

 This survey has focused on a framework of steady-state mechanism design. This approach provides a unified method for analyzing the core trade-offs in settings with stochastic and asynchronous arrivals. However, the broader agenda of dynamic market design  is rich and varied, and many important contributions approach the problem from different angles or with various assumptions. To place the surveyed work in a broader context, we briefly discuss six related streams of literature before turning to promising directions for future research.

\paragraph{Other Perspectives on Queueing and Priority}
The question of optimal queue design has a long history. The framework presented here, which enables the joint optimization of entry, priority, and information, often yields different conclusions than studies that restrict one or more of these design levers. For example, some work has found FCFS to be optimal in models where the stochasticity of the queue length is muted, either by assuming a continuum of agents or specific arrival-departure processes that render the queue size deterministic (e.g., \cite{bloch-cantala2016}; \cite{margaria2020}). Other research focuses on different sources of uncertainty; \cite{cripps-thomas2019}, for instance, analyzes a setting where the service rate is unknown, leading to a problem of strategic experimentation by agents, a different challenge from managing incentives based on the known state of the system. 


\paragraph{Screening with State-Independent Mechanisms}
A significant literature has pioneered the application of mechanism design to queueing systems, particularly for screening agents with heterogeneous preferences. Key contributions include \cite{mendelson1990optimal}, 
\cite{afeche2013incentive}, \cite{afeche2016optimal}, and \cite{kittsteiner2005priority}. These papers study how a service provider can design a menu of options—typically price-priority or price-delay pairs—to induce self-selection. The fundamental difference between this work and the TU framework surveyed here lies in whether the mechanism itself is dynamic. The mechanisms in the aforementioned literature are typically state-independent; the menu of contracts offered to an arriving agent is fixed and does not change with the number of customers already waiting or the system's current state.

\paragraph{Dynamic Matching.}

The problems discussed in this survey is related to the broad domain of dynamic matching, a field that has seen an explosion of recent work (e.g.,
\cite{akbarpour2017thickness}, \cite{akbarpour2020unpaired},
\cite{ashlagi2023matching}, to name just a few).  This literature has often focused on questions of aggregate market performance, such as the value of ``market thickness'' or the optimal timing of batch-matching versus continuous matching. The research highlighted in this survey complements this work by taking a more granular, micro-level mechanism design approach. The focus is less on the aggregate timing of matches and more on the precise design of policies—entry control, priority assignment, and information disclosure—needed to manage the participation and waiting incentives of individual, strategic agents. Also noteworthy is a distinct literature that takes a cooperative games approach to define stability in a dynamic setting; see \cite{damiano2005stability} and \cite{doval2022dynamically} for example. These works grapple with the conceptual challenges of extending static stability to a dynamic context, such as defining credible multi-period blocking plans and formalizing agents' expectations about the future.

\paragraph{Platform Management of Frictional Matching.}  A significant literature studies platforms that manage the inefficiencies inherent in decentralized, frictional matching markets. In these models, agents actively search and incur costs to screen potential partners, which can lead to congestion and wasted effort when they contact others who are ultimately unavailable or uninterested (\cite{arnosti2021managing}; \cite{fradkin2017search}; \cite{horton2019buyer}). Here, the platform's role is not to operate a centralized queue, but to indirectly manage these search frictions by designing the search environment itself. A key lever for the platform is to guide the search process by restricting the set of potential partners an agent might meet. Interventions include designing the search protocol by restricting who can initiate contact (\cite{kanoria2021facilitating}), or by directly setting the meeting rates between different types of agents (\cite{immorlica2021designing}). This approach contrasts sharply with the framework of this survey, which assumes a powerful designer with direct, centralized control over the allocation process via a queue, where frictions primarily manifest as waiting costs within a managed system.

\paragraph{Dynamic Pricing and Revenue Management}
The surveyed works are related to classic literature on dynamic pricing and revenue management (e.g., \cite{gallego1994optimal}; \cite{board2016revenue}; \cite{gershkov2018revenue}; \cite{pai2013optimal}, and \cite{dilme2019revenue}). This literature typically addresses the problem of selling a fixed, perishable inventory (like airline seats) to stochastically arriving buyers over a finite time horizon. In that setting, the dynamics are driven by the non-stationarity of a depleting stock and an approaching deadline. The framework in this survey analyzes a different economic environment: an infinite-horizon system where stochasticity is present on both sides of the market (i.e., supply and demand are both random flows). The central problem is thus the management of long-run, stationary processes, rather than the optimal pricing path for a finite deadline.

\paragraph{Dynamic Mechanism with Evolving Private Information.}
A major branch of the literature on dynamic mechanism design, including works by \cite{courty2000sequential}, \cite{eso2007optimal}, \cite{bergemann2010dynamic}, \cite{athey_efficient_2013},  \cite{pavan_dynamic_2014}, \cite{bergemann2015dynamic}, and \cite{bergemann2022progressive}, addresses a different source of dynamics. In this paradigm, the set of agents is fixed, but their private information—their ``type''—evolves stochastically over time, often as a function of their past allocations (e.g., through learning-by-doing or consuming an experience good). The central design problem is thus not the management of market flows, but the characterization of optimal long-term contracts that provide intertemporal incentives for truth-telling as this information evolves. This focus leads to a different set of analytical tools centered on dynamic versions of the first-order approach, which contrasts with the steady-state and queueing theory methods central to this survey.  See Chapter 11 of  \cite{borgers2015introduction} and \cite{bergemann2019dynamic} for excellent surveys on the subject matter.

\paragraph{Future Directions}

Several promising avenues for research emerge from relaxing the core assumptions of the models discussed.

First, in the NTU setting, a key challenge is to extend the analysis to accommodate heterogeneous agent types in a non-fluid model. The current state of the art for this problem, such as in \cite{ashlagi2025optimal}, largely invokes a fluid or continuum-agent model. While tractable, this approach effectively assumes away the very stochastic frictions and integer-level queue dynamics that make waiting a complex and interesting problem in the first place. Progress in a non-fluid model is challenging because feasible allocation of waiting times among stochastically arriving agents and items is difficult to characterize. 

Second, in the TU setting, while the model in \cite{che2025optimal} allows for heterogeneous values, there is ample scope to incorporate richer forms of heterogeneity, such as private information about agents’ waiting costs, outside options, or their specific service requirements. For example, in many gig platforms, tasks differ not only in the value they create but also in the time they take to complete. Designing mechanisms that can effectively screen agents along these multiple dimensions will require new methodological tools to analyze the interplay between queueing, multi-dimensional screening, and dynamic state-contingent policies.

Both areas represent exciting frontiers for market design. Understanding how to structure dynamic marketplaces---whether through waiting, prices, or a combination of both---remains a central task for economic theory and a practical imperative for the modern digital economy.

\bibliographystyle{economet}
\bibliography{reference}



\renewcommand{\theequation}{\Alph{section}.\arabic{equation}}

\renewcommand{\thetheorem}{\Alph{section}.\arabic{theorem}}

\renewcommand{\theproposition}{\Alph{section}.\arabic{proposition}}

\renewcommand{\thelemma}{\Alph{section}.\arabic{lemma}}
\renewcommand{\theclaim}{A.\arabic{claim}}

\renewcommand{\thecorollary}{\Alph{section}.\arabic{corollary}}

\renewcommand{\thedefinition}{\Alph{section}.\arabic{definition}}

\renewcommand{\theexample}{\Alph{section}.\arabic{example}}

\renewcommand{\thefootnote}{\Alph{section}.\arabic{footnote}}

\renewcommand{\thetable}{\Alph{section}.\arabic{table}}

\renewcommand{\thefigure}{\Alph{section}.\arabic{figure}}


\clearpage 
\appendix  \label{Appendix}
\begin{center}
    {\bf\Large Online Appendix}
\end{center}
\section{Regenerative Processes} \label{app:regen}

Given a stochastic process, which we denote as $\theta=\{\theta(t):t\ge0\}$, we can imagine that there exists a specific, random time $\tau_1$ at which the process probabilistically "starts over". If the evolution of the process from this time onward, $\{\theta(\tau_1+t):t\ge0\}$, has the same distribution as the original process and is independent of its past history, then we say that the process $\theta$ has regenerated at time $\tau_1$. The portion of the process's evolution within the interval $[0, \tau_1)$, along with the regeneration time $\tau_1$ itself, is called the first cycle, denoted $C_1 = \{\{\theta(t):0 \le t < \tau_1\}, \tau_1\}$. The duration of this cycle is the first cycle length, $T_1 = \tau_1$.

If such a regeneration time $\tau_1$ exists, the nature of the process implies that this restarting behavior will continue. There must be a second regeneration time $\tau_2 > \tau_1$, which marks the end of a second cycle, $C_2$, that is identically distributed to the first. Continuing this logic, we find a sequence of regeneration times $\{\tau_k:k \ge 1\}$ (with $\tau_0 = 0$) that constitutes a renewal process. The cycle lengths $T_k = \tau_k - \tau_{k-1}$ for $k \ge 1$ are independent and identically distributed (i.i.d.), and consequently, the cycles $C_k = \{\{\theta(\tau_{k-1}+t):0 \le t < T_k\}, T_k\}$ are also i.i.d. objects.

A regenerative process is said to be \textbf{positive recurrent} if the underlying renewal process of regeneration times is also positive recurrent; this is the case when the expected cycle length is finite and positive, or $0 < \E[T_1] < \infty$. If the expected cycle length is infinite ($\E[T_1] = \infty$), the process is called \textbf{null recurrent}.

Because the regeneration times $\{\tau_k\}$ form a renewal process and the cycles $\{C_k\}$ are i.i.d., one can compute long-run time averages, which will be seen to equal the expected value over a cycle divided by its expected length.  If $\theta$ is a regenerative process, it follows that for any measurable function $f(x)$, the transformed process $f(\theta(t))$ is also a regenerative process that shares the exact same regeneration times. This property allows us to establish a general theorem for the time average of such functions.

The main result is stated precisely in the following theorem.

\begin{theorem} \label{thm:regenerative} If $\theta$ is a positive recurrent regenerative process, and $f=f(x)$ is a function such that $\E[\int_{0}^{T_{1}}|f(\theta(s))|ds] < \infty$, then the following hold:
\begin{itemize}
    \item \textit{The long-run time average of the process converges with probability 1:}
    \[
        \lim_{t\to\infty}\frac{1}{t}\int_{0}^{t}f(\theta(s))ds=\frac{\E[Y]}{\E[T]}
    \]
    \item \textit{The limit of the time-averaged expectation also converges to the same value:}
    \[
        \lim_{t\to\infty}\frac{1}{t}\int_{0}^{t}\E[f(\theta(s))]ds=\frac{\E[Y]}{\E[T]}
    \]
\end{itemize}
Here, $T = T_1$ is the first cycle length and $Y = Y_1 = \int_{0}^{T_{1}}f(\theta(s))ds$ is the value of $f(\theta)$ accumulated during the first cycle.
\end{theorem}

\begin{proof}
Let $N(t) = \max\{j: \tau_j \le t\}$ denote the number of regenerations by time $t$, and let $Y_j = \int_{\tau_{j-1}}^{\tau_{j}} f(\theta(s))ds$ be the reward from the $j^{th}$ cycle. The proof for the first part of the theorem proceeds in two steps.

Assume first the function $f$ is non-negative ($f \ge 0$).  Since time $t$ must fall between the $N(t)$-th and $(N(t)+1)$-th regeneration, we can bound the integral of $f(\theta(s))$. The total reward is at least the sum of rewards from all completed cycles and at most the sum of rewards from all completed cycles plus the reward from the entire next cycle. This gives the sandwich inequality:
\[
\frac{1}{t}\sum_{j=1}^{N(t)}Y_j \le \frac{1}{t}\int_{0}^{t}f(\theta(s))ds \le \frac{1}{t}\sum_{j=1}^{N(t)+1}Y_j
\]
We examine the lower bound first by rewriting it as a product:
\[
\frac{1}{t}\sum_{j=1}^{N(t)}Y_j = \frac{N(t)}{t} \cdot \frac{1}{N(t)}\sum_{j=1}^{N(t)}Y_j
\]
As $t \to \infty$, the term $\frac{N(t)}{t}$ converges to $\frac{1}{\E[T]}$ by the Elementary Renewal Theorem. The second term, being the average of i.i.d. random variables $Y_j$ with a finite mean (as per the theorem's hypothesis), converges to $\E[Y]$ by the Strong Law of Large Numbers. Thus, the lower bound converges to $\frac{\E[Y]}{\E[T]}$.

Similarly, the upper bound can be rewritten as $\frac{N(t)+1}{t} \cdot \frac{1}{N(t)+1}\sum_{j=1}^{N(t)+1}Y_j$. Since $\frac{N(t)+1}{t} \to \frac{1}{\E[T]}$, the upper bound also converges to the same limit $\frac{\E[Y]}{\E[T]}$. By the Squeeze Theorem, the integral term must also converge to $\frac{\E[Y]}{\E[T]}$, which proves the result for the non-negative case. As a direct consequence, the difference between the upper and lower bounds, which is precisely $\frac{Y_{N(t)+1}}{t}$, must converge to 0 with probability 1.

Suppose next that function $f$ need not be non-negative.  We can apply the result from the first case to the non-negative function $|f|$. Let $Y_j^* = \int_{\tau_{j-1}}^{\tau_{j}}|f(\theta(s))|ds$. From the prior step, we know that $\frac{Y_{N(t)+1}^*}{t} \to 0$ with probability 1.

Now, we decompose the integral of $f$ into the sum over complete cycles and the remainder in the last, incomplete cycle:
\[
\frac{1}{t}\int_{0}^{t}f(\theta(s))ds = \frac{1}{t}\sum_{j=1}^{N(t)}Y_j + \frac{1}{t}\int_{\tau_{N(t)}}^{t}f(\theta(s))ds
\]
The first term on the right-hand side converges to $\frac{\E[Y]}{\E[T]}$, as shown before. The second term is the error term, which we must show converges to 0. We can bound its magnitude:
\[
\left|\frac{1}{t}\int_{\tau_{N(t)}}^{t}f(\theta(s))ds\right| \le \frac{1}{t}\int_{\tau_{N(t)}}^{t}|f(\theta(s))|ds \le \frac{1}{t}\int_{\tau_{N(t)}}^{\tau_{N(t)+1}}|f(\theta(s))|ds = \frac{Y_{N(t)+1}^*}{t}
\]
Since we established that $\frac{Y_{N(t)+1}^*}{t} \to 0$, the error term converges to 0. This completes the proof for the first part of the theorem.
\end{proof}

A straightforward application of  \Cref{thm:regenerative} characterizes the stationary distribution.

\begin{corollary} Suppose $P$ is the stationary distribution of $\t=\{\t(t): t\ge 0\}$, a positive recurrent regenerative process. Then, for each measurable set $\T'\subset \T$,
\[
\Pr\{\theta\in \T'\} = \lim_{t\to\infty}\frac{1}{t}\int_{0}^{t}\mathbf{1}_{\{\theta(s)\in \T'\}}ds=\frac{\E[Y]}{\E[T]}, \quad \text{with probability 1}
\]
and also,
\[
\Pr\{\theta \in \T'\} = \lim_{t\to\infty}\frac{1}{t}\int_{0}^{t}\Pr\{\theta\in \T'\}ds=\frac{\E[Y]}{\E[T]},
\]
where $T=T_1$ and  $Y = \int_{0}^{T_{1}}\mathbf{1}_{\{\theta(s)\in \T'\}}ds$.  Hence, $P$ is also a  unique limiting distribution of $\t(t).$
\end{corollary}

\begin{proof}  We first note that a stationary distribution coincides with a long run time average. We apply \Cref{thm:regenerative} by setting $f(x) = \mathbf{1}_{\{x \in \T'\}}$, which  satisfies the conditions of the theorem. 
\end{proof}

We do not present the limit theorem for a regenerative process, as it requires additional background knowledge (e.g., the key renewal theorem).  An interested reader is referred to \cite{asmussen2003applied}.

\section{Proof of \Cref{thm:hassin}}
\begin{proof}  Suppose the first incumbent decides to leave the queue if and only if the queue reaches a length $K$.(That the optimal decision takes this form is obvious.) We can characterize the value $w_k$ to the first incumbent of the queue length (including herself) being $k=1, ..., K$ via dynamic programming.  First, observe that 
\begin{align*}
    w_1=  (\m dt)v  +(\l dt)w_2+ (1-\m dt -\l dt)(w_1-c dt) +o(dt),
\end{align*}
since during the next short time interval $dt$, either he receives his service and collects $v$ with probability $\m dt$ or a new entrant arrives with probability $\l dt$ and his value shifts to $w_2$, or neither happens with the remaining probability,  in which case he stays at the same state except that he incurs the waiting costs $c dt$.  Similarly, we can write:
\begin{align*}
    w_k &= (\m dt ) w_{k-1} +(\l dt) w_{k+1} +(1-\m dt -\l dt) (w_k-c dt) +o(dt), \forall k=2,..., K-1;\\
   w_K&=  (\m dt)w_{K-1}  +(1-\m dt -\l dt)(w_K-c dt)  +o(dt),
\end{align*}
where we use the fact that when a new entry occurs at $k=K$, the first incumbent immediately exits the queue.

Dividing all equations by  $dt$ and letting $dt\to  0$, we obtain a system of equations:
\begin{align*}
    (\l +\m) w_1&=\m v+ \l w_2 -c;\\
      (\l +\m) w_{k-1}&=\m w_{k-2} +\l w_k-c, \forall k=2,..., K-1; \\
      (\l +\m) w_{K}&=\m w_{K-1} -c.  
\end{align*}
This system has a unique solution. In particular, we focus on its last coordinate:
$$w_K=p_0^K\left(v-\sum_{j=0}^{K-1} \frac{(K-j)c}{\m} \r^{j}\right),$$
for the first incumbent stays in the queue up to length $K$  if and only if $w_K\ge 0$.  Namely, the maximal queue length $K_{\lcfs}$ under LCFS satisfies  $w_{K_{\lcfs}}\ge 0\ge w_{K_{\lcfs}+1}$.  

Define the  social welfare under the cutoff structure with the maximal queue length as 
$$W(K)=\sum_{k=1}^K p_k^K(\m v-ck).$$
A straightforward calculation confirms that 
\begin{align*}
\Delta W(K)&:=W(K)-W(K-1)\\
&\propto v-\sum_{j=0}^{K-1}\frac{(K-j+1)c}{\m} \r^j\\
&\propto w_K.    
\end{align*}
Hence, $w_K\ge 0$ if  and only  if $\Delta W(K)\ge 0$.  Together with the fact that $W(K)$ is quasi-concave in $K$, the conclusion follows. 
\end{proof}

\section{Proof of \Cref{thm:leshno}}
\begin{proof}  We can write the expected waiting time when one finds herself in the queue  of  $k$, or joining the queue of $k$:\footnote{They are the recursive equations applying the logic of dynamic programming, where we use the fact that unless state changes (when a new buyer arrives with $k<K$ or a good arrives), the state remains the same with elapse  of waiting time by $dt$, for a brief period $dt$.}
\begin{align*}
    \tau_k &=  (\m dt) \left[\frac{k-1}{k}\tau_{k-1}\right]+ (\l dt)\tau_{k+1}+(1-\m dt-\l  dt)[\tau_k+dt]+o(dt), \forall k<K;\\
     \tau_K &=  (\m dt) \left[\frac{K-1}{K}\tau_{K-1}\right]+(1-\m dt)[\tau_K+dt]+o(dt),
\end{align*}
where $\tau_{0}\equiv 0.$
Dividing both sides by $dt$ and letting $dt\to 0$, we get
\begin{align}
  (\m+\l) \tau_k &=  \m  \left[\frac{k-1}{k} \tau_{k-1}\right]+\l \tau_{k+1} +1, \, \forall k<K; \label{eq:SIRO-k} \\
    \m \tau_K &= \m \left[\frac{K-1}{K}\tau_{K-1}\right]+1. \label{eq:SIRO-K}
\end{align}
The online appendix of \cite{che2021optimal} (which considers more general primitive processes) shows that the system admits a unique solution which satisfies $$\frac{1}{\m}\le \tau_1\le ....\le \tau_K\le \frac{K}{\m}.$$
In fact, we can show that $K<\frac{K}{\m}$; or else \cref{eq:SIRO-K} implies that $\tau_K=\tau_{K-1}=K/\m$, which, when applied to the penultimate equation, yields $\tau_{K-2}=\frac{(K-1)^2}{K\mu}>\frac{K}{\m}=\tau_K$, a contradiction to the monotonicity. (Analogously, one can show that $\tau_1>1/\m$.) Since $K_{SIR0}$ is the largest $K$ such that $v-c\tau_K \ge 0$ and $K<\frac{K}{\m}$, the $K$-th entrant's wait time under FCFS, we conclude that $K_{\siro}\ge K_{\fcfs}$.  To show $K_{\siro}\le K^*$, we can assume $K^*<\infty$ without loss.  In the case, \cref{IRB} is binding at $K^*$.  Clearly, \cref{IRB} is satisfied under SIRO.  To see this, write:
\begin{align*}
  &\sum_{k=1}^{K_{\siro}}  p_k^{K_{\siro}}(\m v-  ck)\\
 =& \m v \sum_{k=1}^{K_{\siro}}p_k^{K_{\siro}}  - c \sum_{k=1}^{K_{\siro}}p_k^{K_{\siro}}k\\
 =& \l v \sum_{k=1}^{K_{\siro}}p_{k-1}^{K_{\siro}}  - c \l \sum_{k=1}^{K_{\siro}}p_{k_1}^{K_{\siro}}\tau_k\\
 =&\l \sum_{k=1}^{K_{\siro}}p_{k-1}^{K_{\siro}}(v-c\tau_k)\\
 \ge & 0,
\end{align*}
 where the third equation follows from the balance condition associated with the stationarity of $\bf$ (first term) and Little's law (second term), and the last inequality follows from the fact that, for the buyer to enter a queue with $k-1$, $v-c\tau_k\ge 0$. Since \cref{IRB} is binding at $K^*$, it follows that $K^*\ge K_{\siro}$.  
\end{proof}

\section{Proof of \Cref{thm:che-tercieux}} \label{app:che-tercieux}

\begin{proof} 
Recall that the first-best solution corresponds to a cutoff policy in which the designer invites buyers into the queue if and only if the queue state is $k<K^*$ (with a possible rationing at $k=K^*-1$) and asks those invited to stay until they collect $v$. Assume for ease of exposition that there is no rationing at $k=K^*-1$.\footnote{The proof works more generally with a small adjustment in the updating formula; see \cite{che2021optimal} for details.}
The optimal information policy, which, as noted before, informs buyers either ``in'' or ``no entry.'' (The no entry is enforced with a commitment not to allocate the good in case of entry.)   Buyers, in turn, make an inference on the relevant history $h\in \H_0$, based on the designer's mechanism, their recommended actions, as well as the amounts of time $t\ge 0$ they have spent in the queue.  

Suppose a buyer has just arrived and receives the invitation to join the queue. What does he believe about the queue length if he joins the queue?
Given the optimal information policy and the optimal cutoff $K^*$, 
Her rational belief at the steady state will be that after joining the queue, its length will be $k$  with probability:\footnote{\label{fn:pasta} This is simply a  Bayesian update of the stationary distribution of the queue state conditional on the recommendation to join (which indicates that $k\in \{0,..., K^*-1\}$). The use of stationarity means that the buyer has a long-run belief that the system has been running for a long time. This is formally verified as PASTA (``Poisson Arrival Sees Time Averages''); see \cite{wolff1982poisson}.}  
\begin{equation}
{\gamma}_{k}^{0}=\left\{%
\begin{array}{ll}
\frac{p_{k-1}^{K^*}}{\sum_{i=1}^{K^*}p
_{i-1}^{K^*}}= \frac{\r^{k-1}}{\sum_{i=1}^{K^*}\r^{i-1}} & \mbox{ if } k=1,..., K^* \\
\label{eq:belief0} 0 & \mbox{ if } k >K^*. \cr%
\end{array}
\right.
\end{equation}
  By Little's  law, the buyer's expected wait time under the optimal information policy is given by:
$$\frac{\sum_{k=1}^{K^*}k p_k^{K^*} }{\l \sum_{j=1}^{K^*} p_{j-1}^{K^*}},$$  
independently of the service priority rule, if one were to stay in the queue until he collects $v$.  Hence, a buyer's expected payoff from joining the queue under the latter assumption is:
\begin{align*}
v- c \frac{\sum_{k=1}^{K^*}k p_k^{K^*} }{\l \sum_{j=1}^{K^*} p_{j-1}^{K^*}}& =  \frac{1}{ \l \sum_{j=1}^{K^*} p_{j-1}^{K^*}}\left [\l \sum_{j=1}^{K^*} p_{j-1}^{K^*}v-c\sum_{k=1}^{K^*}k p_k^{K^*} \right]\\
& =  \frac{1}{\l \sum_{j=1}^{K^*} p_{j-1}^{K^*}}\left [\m \sum_{k=1}^{K^*} p_k^{K^*}v-c\sum_{k=1}^{K^*}k p_k^{K^*} \right]\\
&\propto \sum_{k=1}^{K^*} p_k^{K^*}(\m v-ck )\ge 0,
\end{align*}
where we use the balance condition to obtain the second equality and use \cref{IRB} to obtain the inequality.  It follows that buyers have an incentive to join the queue (under any service priority rule, including FCFS) when  recommended by the optimal mechanism under the optimal information policy, under the assumption that he will continue to stay in the queue once he joins the queue. 

Next, we show that a buyer has the incentive to stay in the queue once he has joined the queue under FCFS. To this end, we show that the residual expected wait time does not increase in the amount of time $t$ that the buyer spends in the queue. Under FCFS, a sufficient statistic for the latter is a buyer's {\it queue position}, $\ell$, i.e., his arrival order within the queue. At $t=0$, his belief on queue position $\ell=1,..., K^*$ is simply given by  $\g_{\ell}^0$.  As $t$ increases, a buyer's belief about his queue position evolves according to the recursion equation: for any $dt>0$,\footnote{The numerator is the probability that his queue position is $\ell$ after
staying in the queue for a length $t+dt$ of time. This event occurs if either
(i) the buyer already has position $\ell$ in the queue at $t$ and none
of the agents ahead of him or himself have been served during   $dt$; or
(ii) if he has position $\ell+1$ at $t$ and one buyer ahead of him is served
by $t+dt$.}
\begin{equation*}
{\gamma}_{\ell }^{t+dt}=\frac{{\gamma}_{\ell
}^{t}(1-\m dt)+{\gamma}_{\ell
+1}^{t}\m dt }{\sum_{i=1}^{K^* }{ \gamma}%
_{i}^{t}(1- \m dt)} +o(dt),
\end{equation*}
for $\ell=1, ..., K^*$, where $\g^t_{K^*+1}\equiv 0$.
We wish to show that $(\g^0_{\ell})$ likelihood-ratio dominates $(\g^t_{\ell})$, for all $t>0$, which will imply that the expected residual wait time decreases in $t$.  

 One can use the above equations (with $dt\to 0$) to derive a system of ordinary differential equations (ODEs) on the likelihood ratios:
\begin{equation}  \label{eq:lr}
\dot r_{\ell}^{t}= \mu r^{t}_{\ell} \left( r^{t}_{\ell+1} - r_{\ell}^{t}\right),
\end{equation}
for $\ell=2,..., K^*$, where $r^{t}_{K^*+1}\equiv 0$.   The system has a unique solution by appealing to the Banach fixed-point theorem.  \cref{eq:belief0}  yields boundary conditions:  $r_{\ell}^{0}= \r$ for $\ell=2,..., K^*$.  There are two cases. Suppose first $K^*=\infty$.  In this case, $r^t_{\ell}$ is constant in $t$, so trivially we have $r^t_{\ell}\le  r^0_{\ell}$ for all $t$.

Next, $K^*<\infty$.  In this case, $\dot r_{\ell}^0=0$ for all $\ell=2,..., K^*-1$, and $\dot r^0_{K^*}<0$. Differentiating the ODE once more, we get
\begin{equation}  \label{eq:ode2}
\ddot r_{\ell}^{t}= \mu \dot r^{t}_{\ell} \left( r^{t}_{\ell+1} - r_{\ell}^{t}\right)+ \m  r^{t}_{\ell} \left(\dot r^{t}_{\ell+1} - \dot r_{\ell}^{t}\right).
\end{equation}
If $\dot r^t_\ell>0$ at some  $t$, there exists the smallest $t$ at which $\dot r^t_\ell$ crosses zero with $\ddot r_{\ell}^t>0$, for some $\ell$.\footnote{We choose the largest $\ell$ if this happens for multiple $\ell$'s.}  From \cref{eq:ode2}, we have  $\ddot r^t_{\ell}=0$, a contradiction.
\end{proof}

\end{document}